\documentclass[twoside]{IEEEtran} 

\usepackage{amsmath, amssymb, amsthm}
\usepackage{epsfig,pstricks,pgf,color,psfrag}
\usepackage{cite, mathtools,enumerate}

\usepackage{hyperref}

\newtheorem{theorem}{Theorem}
\newtheorem{lemma}{Lemma}
\newtheorem{cor}{Corollary}

\theoremstyle{definition}
\newtheorem{definition}{Definition}

\newtheorem{remark}{Remark}

\newcommand{\snr}{\textnormal{\small  \fontfamily{phv}\selectfont snr}}

\newcommand{\snrs}{\textnormal{\footnotesize \fontfamily{phv}\selectfont snr}}

\newcommand{\bE}{\mathbb{E}}

\newcommand{\bY}{\mathbf{Y}}
\newcommand{\bX}{\mathbf{X}}
\newcommand{\bW}{\mathbf{W}}
\newcommand{\bS}{\mathbf{S}}
\newcommand{\bM}{\mathbf{M}}
\newcommand{\by}{\mathbf{y}}
\newcommand{\bx}{\mathbf{x}}
\newcommand{\bw}{\mathbf{w}}

\newcommand{\bv}{\mathbf{v}}
\newcommand{\bV}{\mathbf{V}}

\newcommand{\bZ}{\mathbf{Z}}
\newcommand{\bA}{\mathbf{A}}

\newcommand{\one}{\boldsymbol{1}}
\newcommand{\rw}{\rightarrow}

\title{Approximate Sparsity Pattern Recovery: Information-Theoretic Lower Bounds}
\author{Galen Reeves, \IEEEmembership{Member,~IEEE} and Michael Gastpar \IEEEmembership{Member,~IEEE}

\thanks{Material in this paper was presented in part at the IEEE International Symposium on Information Theory, Toronto, Canada, July 2008 and the 46th Annual Conference on Information Sciences and Systems, Princeton, NJ, March 2012.}
\thanks{This work was supported in part by ARO MURI No. W911NF-06-1-0076.}
\thanks{G. Reeves was with the Department of Electrical Engineering and Computer
Sciences, University of California, Berkeley, CA 94720 USA. He is now with the Department of Statistics, Stanford University, Stanford, CA 94305-4065 USA. (e-mail: greeves@stanford.edu)}
\thanks{M. Gastpar is with the School of Computer and Communication Sciences,
Ecole Polytechnique F\'ed\'erale (EPFL), 1015 Lausanne, Switzerland,
and with the Department of Electrical Engineering and Computer
Sciences, University of California, Berkeley, CA 94720 USA (e-mail: michael.gastpar@epfl.ch).}
}

\begin{document}
\maketitle

\begin{abstract}
Recovery of the sparsity pattern (or support) of an unknown sparse vector from a small number of noisy linear measurements is an important problem in compressed sensing. In this paper, the high-dimensional setting is considered. It is shown that if the measurement rate and per-sample signal-to-noise ratio (SNR) are finite constants independent of the length of the vector, then the optimal sparsity pattern estimate will have a constant fraction of errors. Lower bounds on the measurement rate needed to attain a desired fraction of errors are given in terms of the SNR and various key parameters of the unknown vector. The tightness of the bounds in a scaling sense, as a function of the SNR and the fraction of errors, is established by comparison with existing achievable bounds. Near optimality is shown for a wide variety of practically motivated signal models.
\end{abstract}



\section{Introduction} \label{sec:intro}

\IEEEPARstart{S}{uppose} that a vector $\bx$ of length $n$ is known to have a small number $k$ of nonzero entries, but the values and locations of the nonzero entries are unknown and must be estimated from a set of $m$ noisy linear projections (or samples) given by the vector
\begin{align}\label{eq:model}
\by = A \bx + \mathbf{\bw}
\end{align}
where $A$ is a known $m \times n$ measurement matrix and $\mathbf{\bw}$ is additive white Gaussian noise. The problem of {\em sparsity pattern recovery} is to determine which entries in $\bx$ are nonzero. This problem, which is known variously throughout the literature as support recovery or model selection, has applications in compressed sensing \cite{Donoho_IT06,CandesRombergTao06A, CandesTao06}, sparse approximation \cite{DevoLore93}, signal denoising \cite{Chen_JSC98}, subset selection in regression \cite{Miller90}, and structure estimation in graphical models \cite{MeinBuhl06}.

A great deal of previous work \cite{MeinBuhl06, ZhaoYu_JMLR06, Wainwright_SharpThresholds_IEEE09, Wainwright_InfoLimits_IEEE09,FLetcher_IEEE09, Wang_IEEE10, akcakaya_IEEE10, AerSalZha10, Reeves_masters, RG08}, has focused on necessary and sufficient conditions for exact recovery of the sparsity pattern. By contrast, this paper studies the tradeoff between the number of samples $m$ and the number of detection errors. We focus on the high-dimensional setting where the sparsity rate (i.e.~the fraction of nonzero entries) and the per-sample signal-to-noise ratio (SNR) are finite constants, independent of the vector length $n$. Our results are information-theoretic lower bounds on the sampling rate $\rho = m/n$ needed to attain a desired detection error rate $D$. These bounds are fundamental in the sense that they hold for any possible recovery algorithm. Our results are given explicitly in terms of the sparsity rate, the SNR, and various key properties of the unknown vector. Complementary upper bounds corresponding to a variety of recovery algorithms are given in the companion paper \cite{RG_IEEE_12}. 

\subsection{Overview of Main Contributions}

We study the high-dimensional setting where the measurement matrix $A$ is generated randomly and independently of the vector $\bx$ and the measurements are corrupted by additive white Gaussian noise. Three contributions of the paper are the following:
\begin{enumerate}
\item {\em Fundamental limits:} We derive lower bounds on the sampling rate needed for optimal recovery algorithms. While previous work has focused on exact recovery \cite{MeinBuhl06, ZhaoYu_JMLR06, Wainwright_SharpThresholds_IEEE09, Wainwright_InfoLimits_IEEE09,FLetcher_IEEE09, Wang_IEEE10} or the scaling behavior for approximate recovery \cite{akcakaya_IEEE10}, our work gives an explicit bound on the tradeoff between the sampling rate and the fraction of detection errors. In conjunction with the upper bounds in \cite{RG_IEEE_12}, this bound provides a sharp characterization between what can and cannot be recovered in the presence of noise. This characterization is rigorous and thus validates recent predictions made using the powerful but heuristic replica method from statistical physics \cite{Tanaka02,  Muller03, GuoVer05,KabWadTan09, RanFleGoy09, GuoBarSha09}. 

\item {\em Insights into the difficulty of recovery:} Using tools from information theory, we find a sharp separation into two problem regimes -- one in which the problem is fundamentally noise-limited, and one in which the problem is limited by the behavior of the sparse components themselves.

\item {\em Effects of prior information:} The upper bounds in \cite{RG_IEEE_12} correspond to settings where the approximate number of nonzero entries is known. By contrast, the lower bounds in this paper apply to settings where the recovery algorithm knows the exact number and distribution of the nonzero entries. Interestingly, the resulting bounds show that in many cases, this additional knowledge does not significantly improve the ability to estimate the sparsity pattern.

\end{enumerate}

Beyond these results, our framework also permits us to prove some further insights. For instance, we provide a tight characterization of both the low-distortion and high-SNR behaviors of the sampling rate for a variety of signal classes.

\subsection{Relation to Previous Work}

A great deal of previous work has focused directly on the fundamental limits of exact sparsity pattern recovery \cite{Gastpar_ISIT00, MeinBuhl06, ZhaoYu_JMLR06, Wainwright_SharpThresholds_IEEE09, Wainwright_InfoLimits_IEEE09,FLetcher_IEEE09, Wang_IEEE10}. An initial necessary bound based on Fano's inequality was provided by Gastpar \cite{Gastpar_ISIT00} who considered Gaussian signals and deterministic measurement matrices. Necessary and sufficient scalings of $(n,k,m)$ were given by Wainwright \cite{Wainwright_InfoLimits_IEEE09} who considered deterministic vectors, characterized by the size of their smallest nonzero elements, and i.i.d.~Gaussian measurement matrices. Wainwright's necessary bound was strengthened in our earlier work \cite{Reeves_masters}, for the special case where $k$ scales proportionally with $n$, and for general scalings by Fletcher et al. \cite{FLetcher_IEEE09} and Wang et al. \cite{Wang_IEEE10}. 

Based on the work outlined above, it is now well understood that $m = \Theta( k \log n)$ samples are both necessary and sufficient for exact recovery when the SNR is finite and there exists a fixed lower bound on the magnitude of the smallest nonzero elements \cite{Wainwright_InfoLimits_IEEE09,FLetcher_IEEE09, Wang_IEEE10}. In contrast to the scaling required for bounded MSE, this scaling says that the ratio $m/k$ must grow without bound as the vector length $n$ becomes large. As a consequence, exact recovery is impossible in the setting considered in this paper, where the sparsity rate, sampling rate, and SNR are finite constants, independent of the vector length $n$.

From an information-theoretic perspective, a number of works have studied the rate-distortion behavior of sparse sources \cite{Weid00, Fletcher_JASP06,Sarvotham_Allerton06, Fletcher_SSP07,Fletcher_SP08,WeidVett08}. Most closely related to this paper, however, is work that has addressed approximate sparsity pattern recovery directly. 
For the special case where the values of the nonzero entries are identical and known (throughout the system), Aeron et al. \cite[Theorem V-2]{AerSalZha10} showed that $m = C \cdot k \log(n/k)$ samples are necessary and sufficient for an ML decoder where the constant $C$ is bounded explicitly in terms of the SNR and the desired detection error rate. In the general setting where the nonzero values are unknown, Akcakaya and Tarokh \cite{akcakaya_IEEE10} showed that $m = C \cdot k \log(n/k)$ samples are necessary and sufficient for a joint typicality recovery algorithm where the constant $C$ is finite, but otherwise unspecified. (It can also be shown that this same result is implied directly by the previous work of Cand\`{e}s et al. \cite{ CandesRombergTao06}.) An important difference between these previous results and the results in this paper is that we give an explicit and relatively tight characterization of the constant $C$ for a broad class of signal models.

\subsection{Notation}
When possible, we use the following conventions: a random variable $X$ is denoted using uppercase and its realization $x$ is denoted using lowercase; a random vector $\bV$ is denoted using boldface uppercase and its realization $\bv$ is denoted using boldface lowercase; and a random matrix $\bM$ is denoted using boldface uppercase and its realization $M$ is denoted using uppercase. We use $[n]$ to denote the set $\{1,2,\cdots,n\}$. For a subset $S \subset [n]$ and vector $\bx$, we use $\bx_S$ to denote the $|S|$-dimensional vector of the entries in $\bx$ indexed by $S$. Also, for any $m \times n$ matrix $A$, we use $A_{S}$ to denote the $m \times |G|$ matrix corresponding to the columns of $A$ indexed by $S$. All logarithms are taken with respect to the natural base. Unspecified constants are denoted by $C$ and are assumed to be positive and finite.

\section{Problem Formulation} \label{sec:stochastic_signal_model}

Throughout this paper, the unknown signal is modeled as a $n$-dimensional random vector $\bX$. We consider a 
noisy linear observation model given by
\begin{align}
\bY = \bA \bX + \frac{1}{\sqrt{\snr}} \bW \label{eq:estimation_problem}
\end{align}
where $\bA$ is a random $m \times n$ matrix, $\snr \in (0,\infty)$ is a fixed scalar, and $\bW \sim \mathcal{N}(0, I_{m \times m})$ is additive white Gaussian noise. The vector $\bX$, matrix $\bA$, and noise $\bW$ are assumed to be mutually independent. Note that if $\bE[ \|\bA \bX\|^2] = m$, then $\snr$ corresponds to the per-sample signal-to-noise ratio of the problem. 

The problem studied in this paper is recovery of the sparsity pattern $S^*$ of $\bX$ which is given by 
\begin{align}
S^* = \{ i \in [n] : X_i \ne 0\}.  
\end{align} 
We assume throughout that a recovery algorithm is given the vector $\bY$, the matrix $\bA$, and the distribution on the vector $\bX$. 
The algorithm then returns an estimate $\hat{S}$ of the true sparsity pattern $S^*$.

\subsection{Distortion Measure}\label{sec:formulation_distortion}

To assess the quality of an estimate $\hat{S}$ it is important to note that there are two types of errors. A {\em missed detection} occurs when an element in $S^*$ is omitted from the estimate $\hat{S}$. The missed detection rate is given by
\begin{align}
\text{MDR}(S^*,\hat{S}) =\frac{1}{|S^*|} \sum_{i=1}^n \one( i \in S^*, i \notin \hat{S}).
\end{align}
Conversely, a {\em false alarm} occurs when an element not present in $S^*$ is included in $\hat{S}$. The false alarm rate is given by
\begin{align}
\text{FAR}(S^*,\hat{S}) =\frac{1}{|\hat{S}|} \sum_{i=1}^n \one( i \notin S^*, i \in \hat{S}).
\end{align}

In general, various tradeoffs between the two errors types can be considered. In this paper, however, we focus exclusively on the distortion measure
\begin{align}
d(S^*,\hat{S}) &= \max\big( \text{MDR}(S^*,\hat{S}) ,\,\text{FAR}(S^*,\hat{S})\big). \label{eq:distortion}
\end{align}
This distortion measure 
is a metric on subsets of $[n]$.

For any distortion $D \in [0,1]$ we define the error probability
\begin{align}
\varepsilon^*_n(D) = \min_{p(\hat{s}|\by,A)} \Pr\left[ d(S^*, \hat{S}) > D \right]\label{eq:error_prob}
\end{align}
where the minimization is over all conditional probability mass functions $p(\hat{s} | \by, A)$ and probability is taken with respect to the distribution on the vector $\bX$, the matrix $\bA$, and the noise $\bW$.

%

\subsection{Signal and Measurement Models}\label{sec:model_assumptions}

To characterize the problem of sparsity pattern recovery, we analyze a sequence of recovery problems $\{ \bX(n), \bA(n), \bW(n)\}_{n \ge 1}$ indexed by the vector length $n$. 

\medskip
{\bf \noindent Stochastic Signal Assumptions:} We consider the following assumptions on a sequence of random vectors $\bX(n) \in \mathbb{R}^n$. 
\begin{enumerate}

\item[SS1] {\em Linear Sparsity:} The sparsity pattern $S^*$ is distributed uniformly over all subsets of $[n]$ of size $k(n)$ where $k(n)$ is a known sequence that obeys
\begin{align}
\lim_{n \rw \infty} \frac{k(n)}{n} = \kappa
\end{align} 
for some {\em sparsity rate} $\kappa \in (0,1/2)$.  

\item[SS2] {\em IID Nonzero Entries:}  The nonzero entries $\{X_i : i \in \mathcal{S}^*\}$ are i.i.d.~$p_U$ where $p_U$ is a probability distribution on the real line with no mass at 0, i.e.~$p_U(\{0\})= 0$

\end{enumerate}

Assumption SS1 says that all but a fraction $\kappa$ of the entries are equal to zero, and Assumption SS2 characterizes the behavior of the nonzero entries. Note that under these assumptions the number of nonzero value $k(n)$ is a deterministic (non-random) property of the distribution on $\bX$, and thus knowledge of the distribution on $\bX$ implies that the exact number of nonzero entries is known.

Throughout the paper, we also use $p_X$ to denote the probability distribution on the real line given by $$p_X = (1-\kappa) \delta_0  + \kappa\,  p_U$$
where $\delta_0$ denotes a point-mass at $x=0$. Note that there is a one-to-one correspondence between the pair $(\kappa,p_U)$ and the distribution $p_X$, and that $p_X$ characterizes the {\em marginal distribution} of each entry in $\bX$.

\medskip
{\bf \noindent Measurement Assumptions:} We consider a subset of the following assumptions on the sequence of measurement matrices  $\bA(n) \in \mathbb{R}^{m(n) \times n}$. 
\begin{enumerate}
\item[M1] {\em Non-Adaptive Measurements:} The distribution on $\bA(n)$ is independent of the vector $\bx(n)$ and the noise $\bW(n)$. 
\item[M2] {\em Finite Sampling Rate:} The number of rows $m(n)$ obeys
\begin{align}
\lim_{n \rw \infty} \frac{m(n)}{n} = \rho
\end{align}
for some {\em sampling rate} $\rho \in (0,\infty)$.
\item[M3] {\em Row Normalization:} The distribution on $\bA(n)$ is normalized such that each of the $m(n)$ rows has unit magnitude on average, i.e. 
\begin{align}
\bE\big[ \| \bA(n)\|^2_F\big] = m(n)
\end{align} 
where $\|\cdot\|_F$ denotes the Frobenius norm. 
\item[M4] {\em IID Entries:} The entries of $\bA(n)$ are i.i.d.~with mean zero and variance $1/n$.
\end{enumerate}

Assumptions M1-M3 are used throughout the paper. A sampling rate $\rho < 1$ corresponds to the {\em compressed} sensing setting where the number of equations $m$ is less than the number of unknown signal values $n$. A sampling rate $\rho =1$ corresponds to the number of linearly independent measurements that are needed to recover an arbitrary vector in the absence of any measurement noise. Assumption M4 is used to provide stronger bounds Section~\ref{sec:iid_bounds}.

\subsection{The Sampling Rate-Distortion Function}

Under Assumptions SS1-SS2 and M1-M3, the asymptotic recovery problem is characterized by the sampling rate $\rho$, limiting distribution $p_X$, and $\snr$.

\begin{definition}\label{def:sampling_rate_distortion_function}
A distortion $D$ is {\em achievable} for a fixed tuple $(\rho,p_X, \snr)$ if there exists a sequence of measurement matrices satisfying Assumptions M1-M3 such that
\begin{align}
\lim_{n \rw \infty}\varepsilon^*_n(D) = 0 \label{eq:achievability}
\end{align}
for a sequence of vectors satisfying Assumptions SS1-SS2.
\end{definition}

\begin{definition}
For a fixed tuple $(D,p_X, \snr)$, the {\em sampling rate-distortion function} $\rho^*(D,p_X, \snr)$ is given by
\begin{align}
\rho^*(D,p_X, \snr)= \inf\{ \rho \ge 0 \, : \, \text{$D$ is achievable}\}.
\end{align}
\end{definition}

To lighten the notation, we will denote the sampling rate-distortion function using $\rho^*$ where the dependence on the tuple  $(D,p_X,\snr)$ is implicit.

\begin{remark}
In \cite{RG_IEEE_12}, upper bounds on the achievable distortion are derived under a related but slightly different set of signal assumptions (e.g. the unknown vector is nonrandom and the recovery algorithm is only given the approximate fraction of nonzero entries). In \cite{reeves_thesis}, it is shown that the lower bounds derived under the assumptions of this paper imply corresponding lower bounds for the setting studied in \cite{RG_IEEE_12}. 
\end{remark}

\section{Bounds for Arbitrary Measurement Matrices}\label{sec:gen_bounds}

This section gives lower bounds on the sampling rate distortion function. These bounds apply generally to any sequence of measurement matrices obeying Assumptions M1-M3.

Before we present our bounds, two more definitions are needed. First, we use the notation
\begin{align}
V_X = \text{Var}(X)
\end{align}
to denote the variance of the distribution $p_X$.

Also, we define
\begin{align}\label{eq:R_D}
R(D;\kappa)=
\begin{cases}
 H(\kappa ) - \kappa  H(D) - (1\!-\!\kappa) H\big(\frac{\kappa D}{1\!-\!\kappa}\big), & \text{$D < 1\!-\!\kappa$}\\
0, & \text{$ D \ge 1\!-\!\kappa$}
\end{cases}
\end{align}
where $H(p) = -p \log p - (1-p) \log (1-p)$ is binary entropy. It is straightforward to show that $R(D;\kappa)$ corresponds to the information rate (given in nats per dimension) required to encode a sparsity pattern to within distortion $D$.

\subsection{Initial bound via Fano's inequality}

We begin with a lower bound on the achievable distortion. This bound is general in the sense that it depends only on the variance $V_X$ of the distribution $p_X$, and it serves as a building block for our stronger bounds. The proof is based on Fano's inequality and is given in Appendix~\ref{sec:thm:LB_gen1_proof}. 

\begin{theorem}\label{thm:LB_gen1}
Under Assumptions SS1-SS2 and M1-M3, a distortion $D$ is not achievable for the tuple $(\rho,p_X,\snr)$ if
\begin{align}
\min(1,\rho) \log\big( 1+ \max(1,\rho) V_X \, \snr\big)  < 2 R(D;\kappa). \label{eq:LB_gen1}
\end{align}
\end{theorem}
\begin{remark}
Results similar to Theorem~\ref{thm:LB_gen1} have been derived previously in the special case of exact recovery \cite{Gastpar_ISIT00, Sarvotham_Allerton06, Wang_IEEE10}, as well as for approximate recovery in the special case of binary signals \cite{AerSalZha10}.
\end{remark}

Using Theorem~\ref{thm:LB_gen1} and the concavity of the logarithm, we obtain a simplified lower bound on the sampling rate-distortion function:
\begin{align}
\rho^* \ge \frac{2 R(D;\kappa)}{\log(1+ V_X \snr)} \label{eq:LB_simple}.
\end{align}

Theorem~\ref{thm:LB_gen1} shows that a nonzero sampling rate $\rho$ is necessary in the presence of noise for any distribution $p_X$ with finite variance. One critical weakness, however,  is that it does not reflect the true difficulty of sparsity pattern recovery when the desired distortion $D$ is small. For example, if $D = 0$, then the corresponding lower bound on sampling rate is finite even though it has been shown that an infinite sampling rate is needed in the presence of noise \cite{Reeves_masters}. Among other things, this discrepancy leaves open the possibility that the total number of recovery errors could grow sublinearly with the length $n$ such that the fraction of errors is asymptotically zero.

\subsection{Improved lower bound via a genie argument}

Our next result allows us to lower bound the distortion corresponding to a distribution $p_X$ in terms of a different but related distribution $p_Z$. This result is useful since it allows us to isolate the key aspects of the recovery problem that make recovery difficult. The proof is based on a genie argument and is given in Appendix~\ref{sec:prop_genie_proof}.

\begin{lemma}\label{prop:genie}
Let $p_X$ and $p_Z$ be probability measures with the following properties:
\begin{align}
&0 < \kappa_Z \le \kappa_X \label{eq:kZ_bound}\\
&\frac{p_Z(A)}{1-\kappa_Z} \le \frac{p_X(A)}{1-\kappa_X} \quad \text{for all $A \subseteq \mathbb{R}\backslash \{0\}$} \label{eq:pZ_bound}
\end{align}
where $\kappa_X = 1 - p_X(\{0\})$ and $\kappa_Z = 1 - p_Z(\{0\})$. 
%
%
For a given tuple $(D,\rho,\snr)$ define
\begin{align}
\tilde{D} &= \Big( \frac{1-\kappa_Z}{1-\kappa_X} \Big) \Big( \frac{\kappa_X}{\kappa_Z} \Big) D \label{eq:tilde_D}\\
\tilde{\rho} & = \Big( \frac{1-\kappa_Z}{1-\kappa_X} \Big) \rho\\
\tilde{\snr} & = \Big( \frac{1-\kappa_Z}{1-\kappa_X} \Big) \snr.
\end{align}
Under Assumptions SS1-SS2 and M1-M3, the following statement holds: If the distortion $\tilde{D}$ is not achievable for the tuple $(\tilde{\rho},p_Z,\tilde{\snr})$, then the distortion $D$ is not achievable for the tuple $(\rho,p_X,\snr)$. 
\end{lemma}

Combining Theorem~\ref{thm:LB_gen1} and Lemma \ref{prop:genie} gives the first main result of this paper.  This result overcomes the weakness of Theorem~\ref{thm:LB_gen1} and characterizes the difficulty of recovery when the desired distortion $D$ is small. 

\begin{theorem}\label{thm:LB_gen2}
Under Assumptions SS1-SS2 and M1-M3, a distortion $D$ is not achievable for the tuple $(\rho,p_X,\snr)$ if there exists a tuple $(\tilde{D}, \tilde{\rho},p_Z,\tilde{\snr})$ satisfying the assumptions of Lemma~\ref{prop:genie} such that
\begin{align}
\min(1,\tilde{\rho}) \log( 1+ \max(1,\tilde{\rho}) V_Z \, \tilde{\snr})  < 2 R(\tilde{D};\kappa_Z). \label{eq:LB_gen2}
\end{align}
\end{theorem}

To understand the implications of Theorem~\ref{thm:LB_gen2} it is useful to consider the following simplification.  First, observe that we can parameterize the pair $(D,\tilde{D})$ in terms of the ratio $D' = D/\tilde{D}$. Next, let $p_Z$ be the distribution that minimizes $\bE[Z^2]$ subject to the constraints \eqref{eq:kZ_bound} and \eqref{eq:pZ_bound} with $\kappa_Z = \kappa_X D'/(1- \kappa_X D')$.  As a simple exercise, it can then be verified that
\begin{align}
V_Z \, \tilde{\snr} = P(D';p_X) \, \snr \label{eq:V_Z_simplified}
\end{align}
where
\begin{align}
P(D';p_X) = \int_0^\infty\! \! \max\!\big( \Pr[X^2 > u] - (1\!-\!D')\kappa,0\big) du. \label{eq:P_D}
\end{align}
The function $P(D';p_X)$ corresponds to the average power of the smallest fraction $D'$ of nonzero entries and has been studied previously in the analysis of maximum likelihood estimation (see \cite{RG_IEEE_12}). It is monotonically increasing in $D'$ with $P(0;p_X) = 0$ and $P(D';p_X) > 0$ for any $D' >0$. 

Finally, using the same simplification that led to \eqref{eq:LB_simple} and maximizing over the ratio $D'$ gives the following result. 

\begin{cor}\label{cor:LB_gen2a}
Under Assumptions SS1-SS3 and M1-M3, the sampling rate-distortion function $\rho^*$ obeys
\begin{align}
\rho^* \ge \max_{D' \in [D,1]} \frac{2 (1 - \kappa +\kappa D' ) R\big( \frac{D}{D'} ; \frac{\kappa D'}{1 -\kappa + \kappa D' } \big) }{\log(1+P(D'; p_X)\, \snr)}. \label{eq:LB_gen2a}
\end{align}
\end{cor}

The next section shows that the right-hand side of \eqref{eq:LB_gen2a} tends to infinity as $D \rw 0$. Therefore, one important contribution of Theorem \ref{thm:LB_gen2} is that it is not possible to have a vanishing fraction of errors if both the sampling rate and SNR are finite. 

\subsection{Low-Distorion Behavior}

We now investigate the low-distortion behavior of Theorem~\ref{thm:LB_gen2}. The following result follows directly from Corollary~\ref{cor:LB_gen2a}. The proof is given in Appendix~\ref{sec:low_D_appendix}. 

\begin{cor}\label{prop:lowD_SNR} Fix any $\alpha > 1$. Under assumptions SS1-SS2 and M1-M3, the sampling rate-distortion function $\rho^*$ obeys
\begin{align}
\liminf_{D \rw 0} \rho^* \cdot  \frac{P(\alpha D;p_X)}{ 2 (\alpha -1)  \kappa D \log(1/D) } \ge \frac{1}{\snr}.
\end{align}
\end{cor}

By the continuity of $P(D;p_X)$ over the interval $D \in [0,1)$, one implication of Corollary~\ref{prop:lowD_SNR} is that for any distribution $p_X$, there exists a positive constant $C$ such that the sampling rate-distortion function obeys
\begin{align}
\rho^* \ge C \cdot \frac{D \log(1/D)}{P(D;p_X) \cdot \snr}.
\end{align}

To characterize limiting behavior of the function $P(D;pX)$ we consider two different signal classes:
\begin{itemize}
\item {\em Bounded: } We use $\mathcal{P}_\text{Bounded}(\kappa,B)$ to denote the class of all distributions $p_X$ with sparsity rate $\kappa$, second moment equal to one, and 
$$\Pr[ |X| < B | X \ne 0] = 0$$
for some {\em lower bound} $B >0$. Due to the second moment constraint, the lower bound $B$ cannot exceed $1/\sqrt{\kappa}$.

\item {\em Polynomial Decay: } We use $\mathcal{P}_\text{Poly.}(\kappa,L,\tau)$ to denote the class of all distributions $p_X$ with sparsity rate $\kappa$, second moment equal to one, and $$\lim_{x \rw 0} \frac{\Pr[ |X| \le x | X \ne 0]}{x^L} = \tau$$ for some {\em polynomial decay rate} $L >0$ and limiting constant $\tau \in (0,\infty)$. 
\end{itemize}

The bounded class corresponds to the setting where the nonzero entries in $\bx$ have a fixed lower bound $B$ on their magnitudes, independent of the vector length $n$. By contrast, the polynomial decay class corresponds to the setting where the magnitude of the $\lceil \beta \,k\rceil$'th smallest nonzero entry is proportional to $\beta^{1/L}$ for small $\beta$. Note that in the case of polynomial decay, a vanishing fraction of the nonzero entries are tending to zero as the vector length $n$ becomes large.

Combining Corollary~\ref{prop:lowD_SNR} with analysis of $P(D;p_X)$ given in \cite{RG_IEEE_12} leads to the following result. The proof is given in Appendix~\ref{sec:low_D_appendix}. 

\begin{cor}\label{cor:lowD}
Under assumptions SS1-SS2 and M1-M3 the sampling rate-distortion function obeys the following asymptotic lower bounds:
\begin{enumerate}[(a)]
\item 
If $p_X \in \mathcal{P}_\text{Bounded}(\kappa,B)$, then
\begin{align}
\liminf_{D \rw 0} \frac{\rho^*}{ \log(1/D)} \ge \frac{2}{B^2 \cdot \snr}.
\end{align}
\item 
If $p_X \in \mathcal{P}_\text{Poly.}(\kappa,L,\tau)$, then
\begin{align}
\liminf_{D\rw 0}  \frac{\rho^*}{D^{2/L} \log(1/D)} &= \Big(\frac{\tau}{1+ L/2} \Big)^{-2/L}  \cdot \frac{2}{\snr}. \label{eq:P_D_poly_decay}
\end{align}
\end{enumerate}
\end{cor}

Simply put, Corollary~\ref{cor:lowD} shows that the sampling rate-distortion function obeys
\begin{align}
\rho^* \ge C \cdot \log(1/D)
\end{align}
if $p_X$ is bounded and
\begin{align}
\rho^* \ge C \cdot D^{2/L} \log(1/D)
\end{align}
if $p_X$ has polynomial decay $L$. In \cite{RG_IEEE_12}, it is shown that, up to constants, these scalings are also achievable. Together, these upper and lower bounds characterize precisely how the sampling rate-distortion function increases as the desired distortion becomes small.

\section{Bounds for IID Measurement Matrices}\label{sec:iid_bounds}

This section gives stronger lower bounds for measurement matrices whose entries are i.i.d.~(Assumption M4). Unlike the bounds given in the previous section, these bounds capture the fact that the values of nonzero entries of $\bX$ are unknown. Section~\ref{sec:continuous_dist} gives an improved lower bound for settings where the nonzero entries are continuous. Section~\ref{sec:high_SNR} considers the high-SNR behavior of the bound. Section~\ref{sec:arbitrary_dist} extends the bound arbitrary distributions. 

\subsection{Improved lower bound via the entropy power inequality}\label{sec:continuous_dist}

We define the nonzero entropy power of a random variable $X \sim p_X$ to be 
\begin{align}
N_X= 
\begin{dcases}
 \frac{ \kappa \exp\left(2 h(X|X \ne 0)\right)}{2 \pi e}, & \text{if $h(X|X \ne 0)$ exists }\\
0, & \text{otherwise}
\end{dcases}
  \label{eq:N_X}
\end{align}
where $h(X|X\ne 0)$ denotes the differential entropy of the nonzero part of $p_X$. The nonzero entropy power allows us to assess the relative uncertainty about the nonzero entries. 

Our next result gives a lower bound on the achievable distortion in terms of the variance $V_X$ and the nonzero entropy power $N_X$. The proof relies heavily on the entropy power inequality and the spectral convergence of i.i.d.~random matrices and is given in Appendix~\ref{sec:thm:LB_iid_1_proof}.

\begin{theorem}\label{thm:LB_iid_1}
Under Assumptions SS1-SS3 and M1-M4, a distortion $D$ is not achievable for the tuple $(\rho,p_X,\snr)$ if\begin{align}\label{eq:LB_iid_1}
\mathcal{V}(\rho,V_X\, \snr) - \kappa \mathcal{V}_\text{LB}(\rho/\kappa,N_X\, \snr) < 2 R(D;\kappa) 
\end{align}
where 
\begin{align}
\mathcal{V}(r,\gamma) &=  r \log\big(1+ \gamma - \mathcal{F}(r,\gamma)\big) 
+ \log\big(1+r\, \gamma - \mathcal{F}(r,\gamma)\big) \nonumber \\
 & \quad -   \mathcal{F}(r,\gamma)/\gamma  \label{eq:V_funct}
\end{align}
with
\begin{align*}
\mathcal{F}(r,\gamma)  = \frac{1}{4} \left( \sqrt{\gamma\, (\sqrt{r} + 1)^2 +1} -  \sqrt{\gamma\, (\sqrt{r}-1)^2 +1} \right)^2
\end{align*}
and
\begin{align}
\mathcal{V}_{LB}(r,\gamma)=
\begin{dcases}
r\log\Big(1+\gamma \Big( \frac{1}{1-r} \Big)^{1/r- 1} \frac{1}{e}\Big),& \text{if $r < 1$}\\
\; \log\Big(1+  \gamma \frac{1}{e} \Big),& \text{if $r=1$}\\
\; \log\Big(1+\gamma \, r  \Big( \frac{r}{r-1} \Big)^{r- 1} \frac{1}{e}\Big),& \text{if $r >1$}
\end{dcases} \label{eq:V_LB_funct}.
\end{align}
\end{theorem}

\begin{remark}
In the special case where the nonzero part of the distribution $p_X$ is Gaussian, the function $\mathcal{V}_\text{LB}(r,\gamma)$ in the second term on the left-hand side of \eqref{eq:LB_iid_1} can be replaced with the function $\mathcal{V}(r,\gamma)$, thus providing a slightly stronger condition. 
\end{remark}

Combining Theorem~\ref{thm:LB_iid_1} with bounds on the functions $\mathcal{V}(r,\gamma)$ and $\mathcal{V}_\text{LB}(r,\gamma)$ given in Appendix~\ref{sec:random_matrix}, gives a simplified lower bound on the sampling rate-distortion function:
\begin{align}
\rho^* > \frac{\min(\rho^*,\kappa) \log(1+(N_X/e)\, \snr ) + 2R(D,\kappa)  }{\log(1+V_X \, \snr) }. \label{eq:rho_LB_iid_simple}
\end{align}
Note that this bound is similar to \eqref{eq:LB_simple}, except that there is an additional term on the right-hand side.

\subsection{High-SNR behavior}\label{sec:high_SNR}
The key improvement of Theorem~\ref{thm:LB_iid_1} is that the lower bound on the distortion remains bounded away from zero for all SNR. To illustrate this point, we first consider the infinite SNR limit of the lower bound on the distortion. Since the achievable distortion is non-increasing in the SNR, this limit gives a valid lower bound for any SNR.

\begin{cor}\label{thm:high_SNR_1}
Under Assumptions SS1-SS3 and M1-M4, a distortion $D$ is not achievable for the tuple $(\rho,p_X,\snr)$ if $\rho < \kappa$ and
\begin{align}
&\rho \log \left( \frac{V_X}{N_X} \right)  + (1-\rho) \log \Big( \frac{1}{1-\rho} \Big) - (\kappa - \rho) \log \Big( \frac{\kappa}{\kappa - \rho} \Big) \nonumber
\\& < 2 R(D;\kappa). \label{eq:inf_SNR_lim}
\end{align}
\end{cor}

The proof of Corollary~\ref{thm:high_SNR_1} follows directly from the infinite SNR limit of Theorem~\ref{thm:LB_iid_1} and is  given in Appendix~\ref{sec:high_SNR_appendix}. In \cite{reeves_thesis}, it is shown that the same result can be obtained by direct analysis of the noiseless setting. 

Using the fact that the left-hand side of \eqref{eq:inf_SNR_lim} is increasing in $\rho$ gives a simple lower bound on the sampling rate-distortion function.
\begin{cor}\label{thm:high_SNR_1b}
 Consider Assumptions SS1-SS3 and M1-M4. If the ratio $N_X/V_X$ is large relative to the desired distortion $D$, i.e. if
\begin{align}
 \kappa \log\Big(\frac{V_X}{N_X}\Big) +  (1-\kappa) \log\Big(\frac{1}{1-\kappa}\Big) < 2 R(D;\kappa)  \label{eq:DinfLim},
\end{align}
 then sampling rate-distortion function obeys $\rho^* \ge \kappa$ for all SNR. 
 \end{cor}

The next result gives a precise characterization of the high-SNR behavior of Theorem~\ref{thm:LB_iid_1}. The proof is given in Appendix~\ref{sec:high_SNR_appendix}

\begin{cor}\label{thm:high_SNR_2}
Under the assumptions of Corollary~\ref{thm:high_SNR_1b},
the sampling rate-distortion function obeys
\begin{align}
&\liminf_{\snrs \rw \infty} (\rho^* - \kappa) \log(\snr) \ge  \nonumber \\
&2 R(D;\kappa) - \kappa \log\Big(\frac{V_X}{N_X}\Big) - (1-\kappa) \log\Big(\frac{1}{1-\kappa}\Big).
\end{align}
\end{cor}

Corollary~\ref{thm:high_SNR_2} shows that under the conditions of Corollary~\ref{thm:high_SNR_1b}, the lower bound on the sampling rate distortion function $\rho^*$ obeys
\begin{align}
\rho^* \ge \kappa + \frac{C}{\log(1+\snr)}
\end{align}
for some positive constant $C$. 

For comparison, it is shown in \cite{RG_IEEE_12} that the sampling rate-distortion function obeys the asymptotic upper bound:
\begin{align}
&\limsup_{\snrs \rw \infty} (\rho^* - \kappa) \log(\snr) \le  2 H(\kappa), \label{eq:high_SNR_UB}
\end{align}
and hence
\begin{align}
\rho^* \le \kappa + \frac{\tilde{C}}{\log(1+\snr)}
\end{align}
for some finite constant $\tilde{C}$. Together, these lower and upper bounds characterize precisely how the sampling rate distortion function $\rho^*$ converges to the sparsity rate $\kappa$ as the SNR increases.

\subsection{Extension to arbitrary distributions}\label{sec:arbitrary_dist}

Combining Theorem~\ref{thm:LB_iid_1} with Lemma~\ref{prop:genie} gives the following result which is the strongest bound in this paper.

\begin{theorem}\label{thm:LB_iid_2}
Under Assumptions SS1-SS3 and M1-M4, a distortion $D$ is not achievable for the tuple $(\rho,p_X,\snr)$ if there exists a tuple $(\tilde{D},\tilde{\rho},p_Z,\tilde{\snr})$ satisfying the assumptions of Lemma~\ref{prop:genie} such that
\begin{align}
\mathcal{V}(\rho,V_Z\, \tilde{\snr}) - \kappa_Z \mathcal{V}_\text{LB}(\rho/\kappa_Z,N_Z\, \tilde{\snr}) < 2 R(\tilde{D};\kappa_Z)  \label{eq:LB_iid_2}.
\end{align}
\end{theorem}

Theorem~\ref{thm:LB_iid_2} has the same low-distortion behavior as Theorem~\ref{thm:LB_gen2}. Furthermore it allows us to extend the high-SNR improvements of Theorem~\ref{thm:LB_iid_1} to arbitrary distributions. 

For example, consider the following result.

\begin{cor}\label{cor:LB_iid_2}
Suppose that $p_X$ can be expressed as
\begin{align}
p_X = (1-\kappa) \,\delta_0 + \omega_c\, p_{X_c} + (\kappa - \omega_c) \, p_{\tilde{X}}
\end{align}
where $X_c$ is continuous with finite differential entropy. Under Assumptions SS1-SS2 and M1-M4, a distortion $D$ is not achievable for the tuple $(\rho, p_X,\snr)$ in the noiseless setting if $\rho < \omega_c$ and \eqref{eq:LB_iid_2} holds for the tuple $(\tilde{D}, \tilde{\rho}, p_Z,\tilde{\snr})$ given by
\begin{align}
\tilde{D} &= \frac{\kappa}{\omega_c} D\\
\tilde{\rho} &= \Big( \frac{1}{1-\kappa + \omega_c}\Big) \rho\\
p_Z &= \Big(\frac{1-\kappa}{1-\kappa + \omega_c}\Big) \delta_0 + \Big(\frac{\omega_c}{1-\kappa + \omega_c}\Big) p_{X_c}  \label{eq:p_Z_cor} \\
\tilde{\snr} &= \Big( \frac{1}{1-\kappa + \omega_c}\Big) \snr.
\end{align}
\end{cor}

Starting with Corollary~\ref{cor:LB_iid_2} and following the same steps that let to Corollary~\ref{thm:high_SNR_2} gives the following high-SNR characterization.

\begin{cor}\label{cor:high_SNR} Consider the assumptions of Corollary~\ref{cor:LB_iid_2} and let
\begin{align}
\Delta &= 2(1-\kappa + \omega_c) R\Big(\frac{\kappa}{\omega_c} D;\frac{\omega_c}{1\!-\!\kappa + \omega_c}\Big) \nonumber \\
&\quad -\omega_c \log \Big( \frac{V_Z}{N_Z}\Big) +  (1-\kappa)  \log \Big( \frac{1-\kappa + \omega_c}{1-\kappa}\Big)
 \end{align}
where $p_Z$ is given by \eqref{eq:p_Z_cor}.
If $\Delta > 0$, then the sampling rate-distortion function obeys
\begin{align}
\liminf_{\snr \rw \infty} (\rho^*- \omega_c)\log(\snr) \ge \Delta.
\end{align}
\end{cor}

Corollary~\ref{cor:high_SNR} shows that if the nonzero entropy power $N_{Z}$ is large relative to the desired distortion $D$, then the sampling rate-distortion function obeys
\begin{align}
\rho^* \ge \omega_c + \frac{C}{\log(1+\snr)}
\end{align}
for some positive constant $C$. This result shows that the high-SNR behavior is dominated by the weight of the continuous part of the distribution on the nonzero entries.

\begin{figure*}[htbp]
\centering
\psfrag{y1}[B]{\small Distortion $\alpha$}
\psfrag{x1}[]{\small Sampling Rate $\rho$} 
\psfrag{t1}[B]{}
\psfrag{Upper Bound [?]}[lb]{\small Upper Bound \cite{RG_IEEE_12}}
\psfrag{Lower Bound}[lb]{\small Lower Bound}
\psfrag{(Theorem 1)}[lT]{\small (Proposition \ref{thm:nec1})}
\psfrag{(Theorem 2)}[lT]{\small (Propositions \ref{thm:nec-G} \&  \ref{thm:nec2})}
\epsfig{file=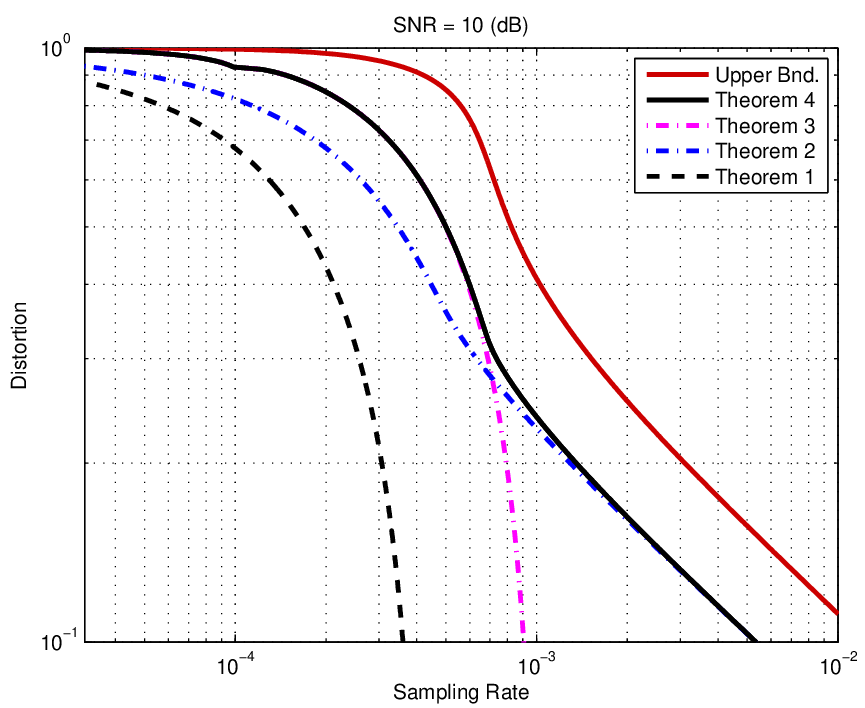, width = 0.45\textwidth}
\epsfig{file=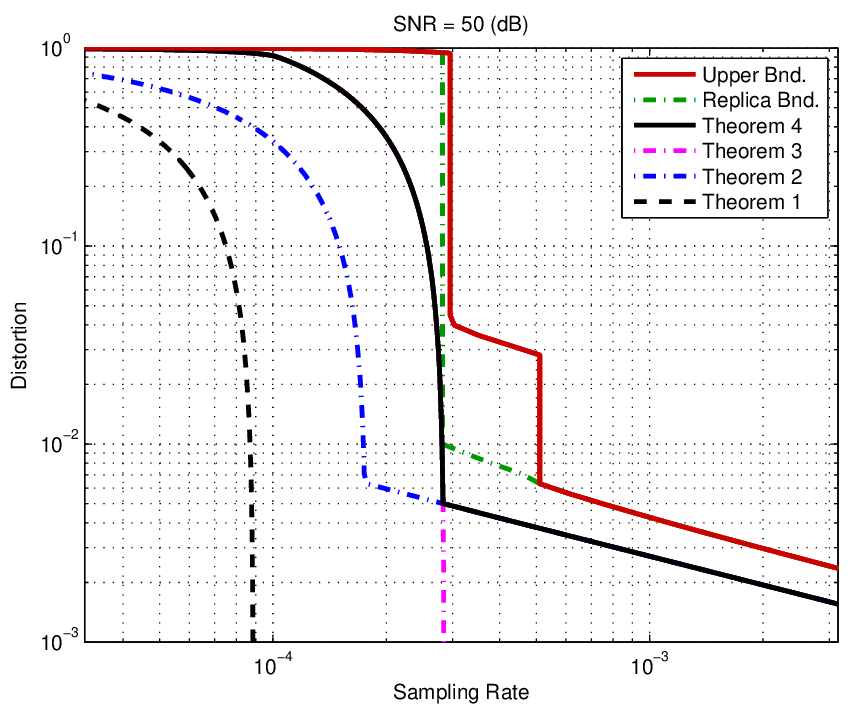, width =0.45\textwidth}
\psfrag{y1}[B]{\small Distortion $\alpha$}
\psfrag{x1}[]{\small Sampling Rate $\rho$} 
\psfrag{t1}[B]{}
\psfrag{Upper Bound [?]}[lb]{\small Upper Bound \cite{RG_IEEE_12}}
\psfrag{Lower Bound}[lb]{\small Lower Bound}
\psfrag{(Theorem 1)}[lT]{\small (Proposition \ref{thm:nec1})}
\psfrag{(Theorem 2)}[lT]{\small (Propositions \ref{thm:nec-G} \&  \ref{thm:nec2})}
\epsfig{file=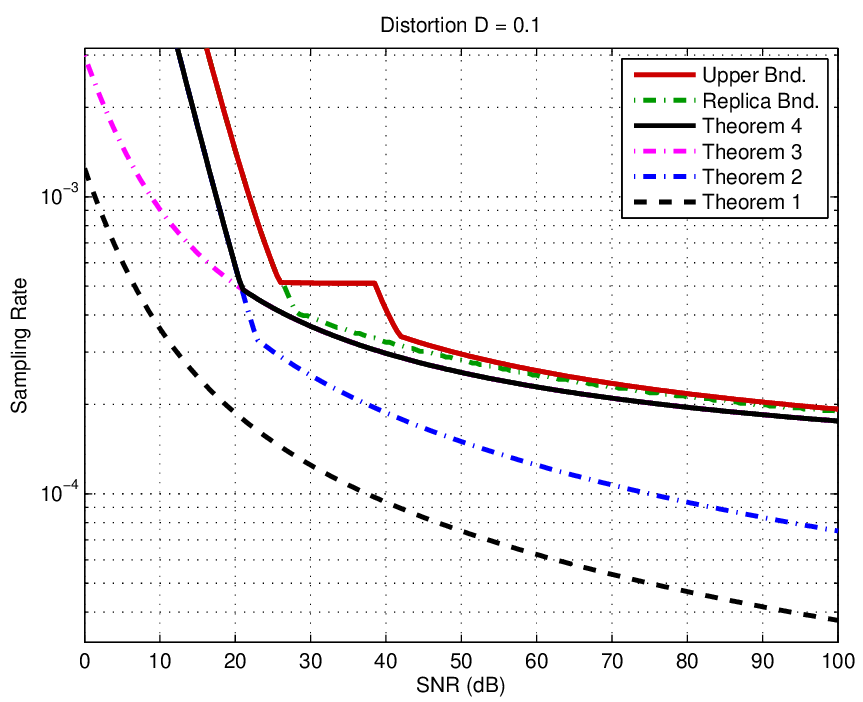, width = 0.45\textwidth}
\epsfig{file=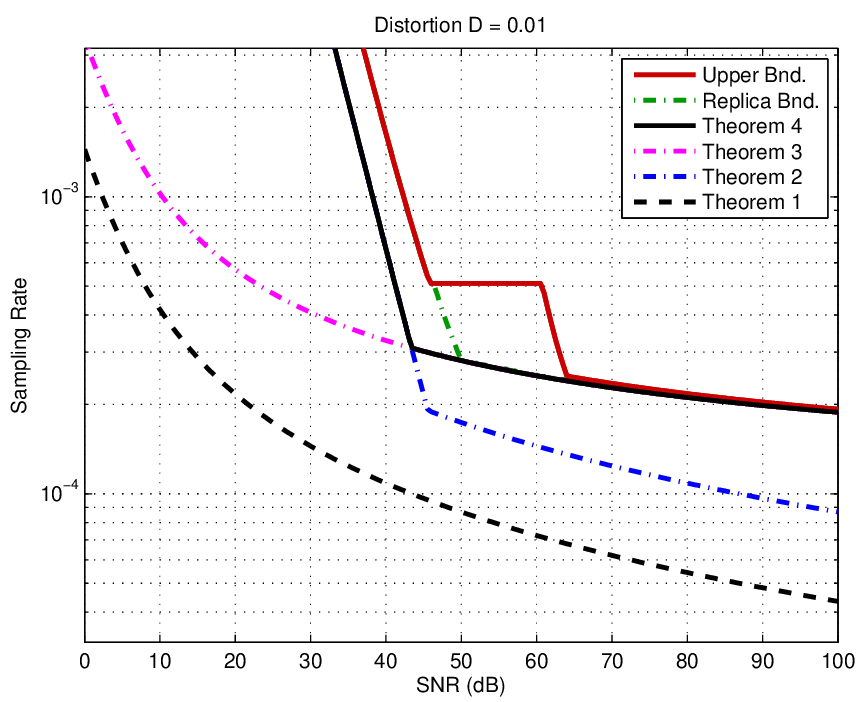, width =0.45\textwidth}
\caption{Comparison of the lower bounds in Theorems 1--4 for a Bernoulli-Gaussian distribution with sparsity rate $\kappa = 10^{-4}$. The top row plots the lower bounds on the distortion $D$ as a function of the sampling rate $\rho$ for fixed SNR. The bottom row plots the lower bounds on the sampling rate $\rho$ as a function the SNR for fixed distortions. Also shown is an upper bound derived in \cite{RG_IEEE_12} and a heuristic bound derived using the standard but nonrigorous replica method from statistical physics (see \cite{RG_IEEE_12} for more details). We remark that the ``kinks'' in the upper bound are likely an artifact of the bounded technique used in \cite{RG_IEEE_12}.}\label{fig:comparison}
\end{figure*}

\section{Examples and Illustrations}\label{sec:analysis} 

This section provides specific examples and illustrations of the bounds developed in this paper.

\subsection{Comparison of Lower Bounds}\label{sec:comparison}

We begin with a comparison of the lower bounds in Theorems~1--4. To illustrate these bounds, we consider the special case of the Bernoulli-Gaussian distribution given by
\begin{align}
X = 
\begin{cases}
0, & \text{with probability $1-\kappa$}\\
W, & \text{ with probability $\kappa$}
\end{cases}
\end{align}
where $W$ is a Gaussian random variable with mean zero and variance $1/\kappa$. This distribution has polynomial decay rate $L=1$ and limiting constant $\tau = \sqrt{2 \kappa /\pi}$. Moreover, its nonzero entropy power $N_X$ is equal to the variance $V_X$. 

The bounds in Theorems 1--4 corresponding to the Bernoulli-Gaussian are shown in Figure~\ref{fig:comparison}. In all cases, the initial lower bound given in Theorem~1 is highly sub-optimal and does not reflect the true difficultly of the recovery problem. By contrast, the strongest bound in this paper, Theorem~4, is in close agreement with the behavior of the upper bound from \cite{RG_IEEE_12}.

The relative merits of Theorems 2 and 3 depend on the problem regime. When the sampling rate is large relative to the SNR and the distortion, the difficulty of recovery is dominated by the magnitude of the smallest nonzero entries and Theorem~2 provides a stronger bound. Conversely, when the sampling rate is small relative to the SNR and distortion, the difficulty of recovery is dominated by the entropy of the nonzero entries and Theorem~3 provides a stronger bound.

The bounds on the achievable distortion plotted in top left panel of Figure~\ref{fig:comparison} show that Theorem~4 can be strictly greater than the maximum of the Theorems 2~and~3.

\begin{figure*}[htbp]
\centering
\epsfig{file=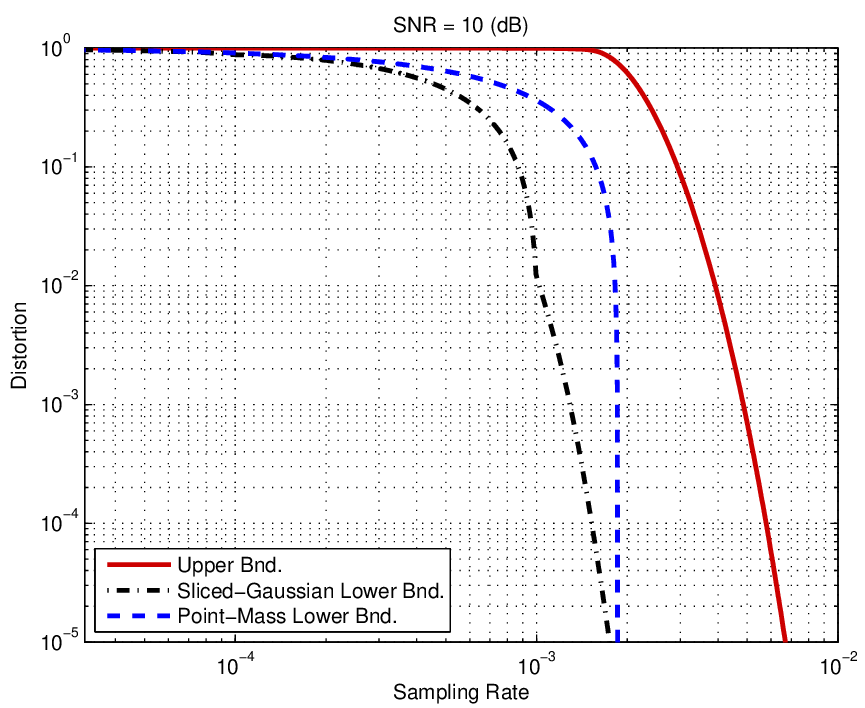, width = 0.45\textwidth}
\epsfig{file=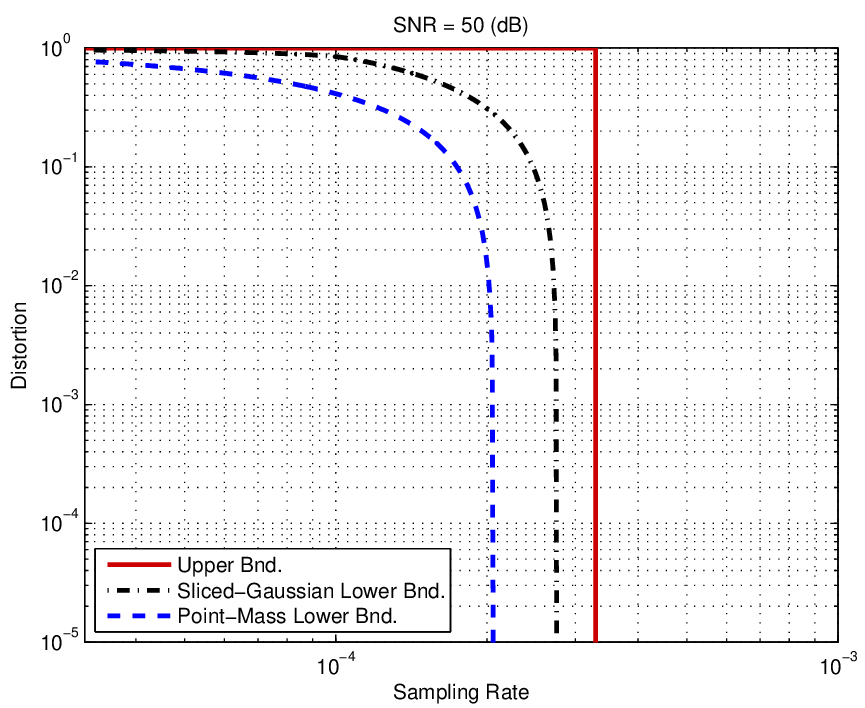, width =0.45\textwidth}
\epsfig{file=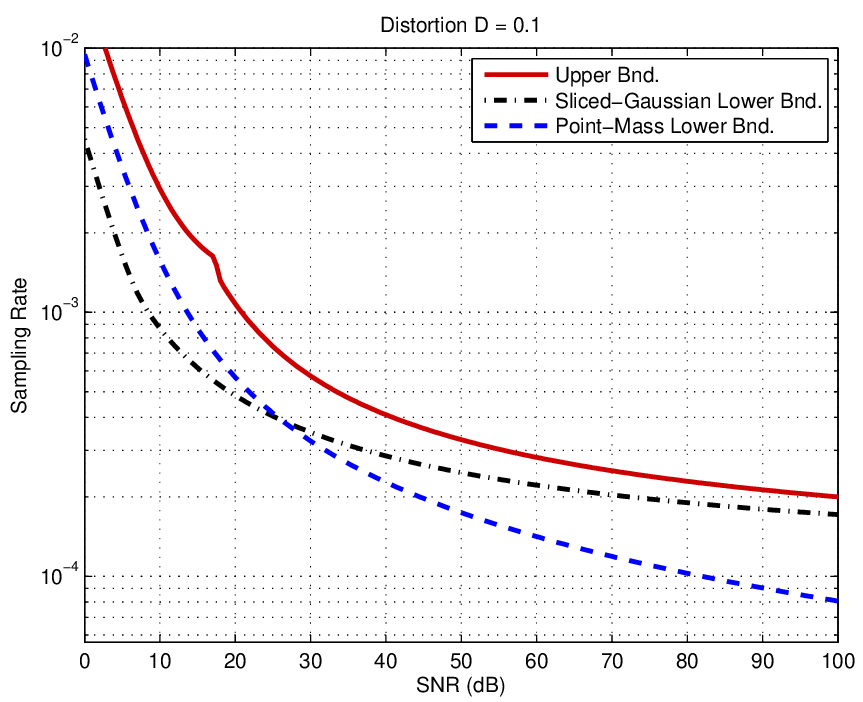, width = 0.45\textwidth}
\epsfig{file=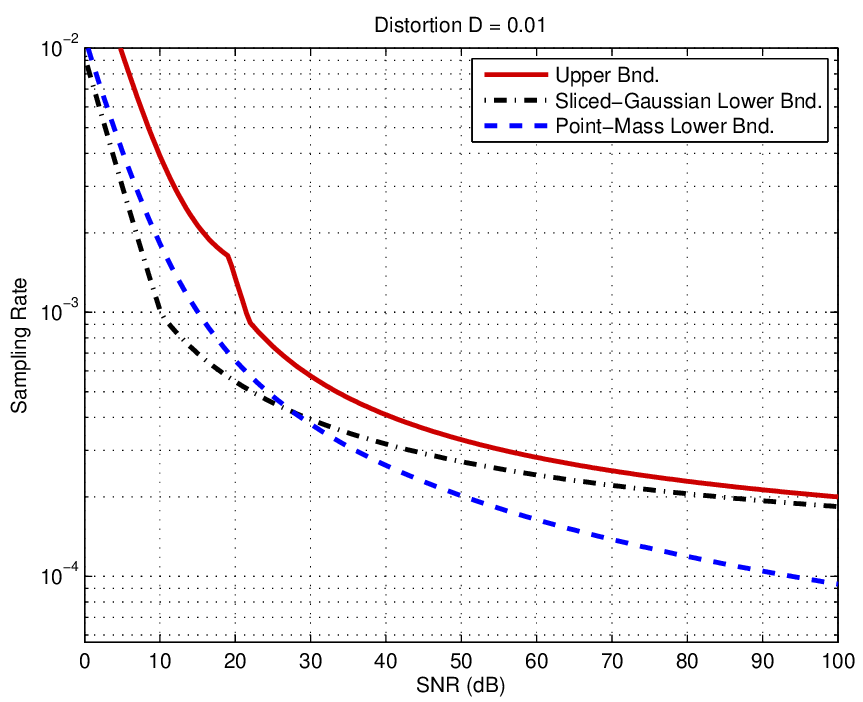, width =0.45\textwidth}
\caption{
Comparison of the lower bound in Theorem 4 for the sliced-Gaussion and point-mass distributions. In both cases, the distributions have second moment equal to one, sparsity rate $\kappa = 10^{-4}$, and lower bound $B = \sqrt{0.2/\kappa}$ (i.e. the nonzero entries of $\bX$ are lower bounded in squared magnitude by 20\% of their average power). The top row plots the lower bounds on the distortion $D$ as a function of the sampling rate $\rho$ for fixed SNR. The bottom row plots the lower bounds on the sampling rate $\rho$ as a function the SNR for fixed distortions. Also shown is a minimax upper bound derived in \cite{RG_IEEE_12} which applies universally over the class of bounded signals $\mathcal{P}_\text{Bounded}(\kappa,B)$.}\label{fig:minimax}
\end{figure*}

\subsection{Lower Bounds for Signal Classes}\label{sec:minimax} 
Throughout this paper, we have assumed that the underlying distribution $p_X$ is known. More realistically though, it may be the case that the distribution $p_X$ is known to belong to class $\mathcal{P}$ of sparse distributions, but is otherwise unknown. In these cases, a distortion $D$ is said to be achievable for a class $\mathcal{P}$ if and only if there exists a fixed estimator $\hat{S}(\bY,\bA)$ such that
\begin{align}
\sup_{p_X \in \mathcal{P}} \Pr[ d(S^*,\hat{S}(\bY,\bA)) > D] \rw 0 \quad \text{as $n \rw \infty$.}
\end{align}

One class of distributions considered widely throughout the literature is the bounded signal class $\mathcal{P}_\text{Bounded}(\kappa,B)$, i.e.~the class of all distributions with sparsity rate $\kappa$, second moment equal to one, and $\Pr[|X| < B| X \ne 0] =0$ for some lower bound $B>0$. In \cite{RG_IEEE_12}, upper bounds on the sampling rate-distoriton function of this class are derived for several different recovery algorithms. In this section, we provide corresponding information-theoretic lower bounds. 

To proceed, we use the simple fact that a distortion is not achievable for a class of distributions if it is not achievable for each distribution $p_X$ in that class. In the following, we evaluate Theorem~4 for two carefully chosen distributions in the class $\mathcal{P}_\text{Bounded}(\kappa,B)$.

\begin{itemize}
\item 
{\em Point-Mass Lower Bound: } For the first bound, we consider the distribution $p_X$ given by
\begin{align*}
\Pr[X=x] = 
\begin{cases}
1-\kappa & \text{if $x = 0$}\\
\kappa\, (1-\epsilon) & \text{if $x = B$}\\
\kappa \, \epsilon & \text{if $x =  \frac{1/\kappa - (1-\epsilon)B^2}{\epsilon}$}
\end{cases}
\end{align*}
for some some $\epsilon \in (0,1)$. If $\epsilon$ is small, then this distribution places almost all of its nonzero mass at the lower bound $B$, and so 
$P(D;p_X) \approx D B^2$ 
for small distortions. Since the distribution is discrete, the nonzero entropy power $N_X$ is equal to zero.

\item
{\em Sliced-Gaussian Lower Bound: } For the second bound, we consider the distribution $p_X$ given by 
\begin{align*}
X = 
\begin{cases}
0, & \text{with probability $1-\kappa$}\\
\text{sgn}(W) B + W, & \text{ with probability $\kappa$}
\end{cases}.
\end{align*}
where $W$ is a Gaussian random variable with mean zero and variance $\sigma_W^2$ scaled so that $\bE[X^2] =1$. For this distribution, the function $P(D;p_X)$ is larger than for the point-mass distribution. However, the nonzero entropy power is $N_X = \kappa\,  \sigma_W^2$.
\end{itemize}

The bounds in Theorem 4 corresponding to the point-mass and sliced-Gaussian distributions are plotted in Figure~\ref{fig:minimax} along with the universal upper bound from \cite{RG_IEEE_12}. We emphasize that the maximum of the two lower bounds is also a valid lower bound for the bounded signal class. 

The relative strengths of two lower bounds depend on the problem regime. When the SNR is large relative to the sampling rate, the sliced-Gaussion distribution provides a stronger bound. Conversely, when the SNR is small relative to the sampling rate, the point-mass distribution provides a stronger bound.

\section{Discussion}

In this section, we review the main contributions of the paper and discuss various implications of our analysis. 

\subsection{Overview of results}
The information-theoretic lower bounds derived in this paper, in conjunction with achievable bounds in \cite{RG_IEEE_12}, characterize the fundamental limit of what cannot be recovered in presence of noise. The results in this paper can be summarized as follows:
\begin{itemize}
\item Theorem~1 gives an initial lower bound based on Fano's inequality. This result, which is closely related to existing bounds in the literature, serves as a building block for the main results.   
\item Theorem~2 gives a significantly improved lower bound based on the genie result given in Lemma 1. In conjunction with the upper bounds in \cite{RG_IEEE_12}, Theorem~2 gives a tight characterization of the low-distortion behavior of the sampling rate-distortion function. 

\item Theorem~3 gives a different lower bound based on the entropy power inequality and the asymptotic spectral convergence of i.i.d.~random matrices. In conjunction with the upper bounds in \cite{RG_IEEE_12}, Theorem~3 gives a tight characterization of the high-SNR behavior of the sampling rate-distortion function for settings where the nonzero entries are continuously distributed.

\item Theorem~4 combines Theorem 3 with the genie result in Lemma~1 to give the strongest lower bound in the paper. This bound combines the low-distortion improvements of Theorem 2 and the high-SNR improvements of Theorem~3.

\end{itemize}

\subsection{Fundamental Behavior of Sparsity Pattern Recovery}

Our bounds show that the tradeoffs between the sampling rate $\rho$, the distortion $D$, and the SNR can be characterized in terms of certain key properties of the underlying distribution $p_X$. The following limiting behaviors are considered.

\subsubsection*{ High-SNR Behavior}
Let the desired distortion $D$ be fixed. As the SNR becomes large, the difficulty of recovery is dominated by the entropy of the nonzero entries. If the nonzero part of $p_X$ has a continuous component with weight $\omega_c$ and a relatively large differential entropy, then the sampling rate-distortion function obeys 
$$\rho^* \ge \omega_c + \frac{C}{\log(\snr)}.$$
This behavior can be seen in the top row of Figure~1. 

\subsubsection*{Low-Distortion Behavior}
Let the SNR be fixed. As the desired distortion becomes small, the difficulty of recovery is dominated by the relative magnitudes of the smallest nonzero entries. If the nonzero entries are bounded away from zero, then the sampling-rate distortion function obeys
\[ \rho^* \ge C \cdot \log(1/D).\]
If the nonzero entries are drawn from a distribution with decay rate $L$, then the sampling rate-distortion function obeys
\[\rho^* \ge C \cdot  D^{2/L} \cdot \log(1/D).\]
This behavior can be seen in the bottom rows of Figures~1 and 2.

\subsection{Role of Model Assumptions}

This paper focuses on the setting where a constant fraction of the entries are nonzero (Assumption SS1). In principle, many of the tools developed in the paper could also be used to address settings where the number of nonzero entries grows sub-linearly with the vector length, and hence there is a vanishing fraction of nonzero entries.

Our use of row normalization (Assumption M3) differs from many related works which use column normalization. The reason for our scaling is that, from a sampling perspective, one way to decrease the effect of noise is to take additional samples (all at a fixed per-measurement SNR). If the column norms of the measurement matrix are constrained, then this is not possible since the per-measurement SNR will necessarily decrease as the number of measurements increases. Since it is assumed throughout that the sampling rate $\rho$ is a fixed constant, all results in this paper can be compared to existing works under an appropriate rescaling of the SNR.

The proofs of Theorems 3 and 4 rely heavily on the assumption that the measurement matrices have i.i.d.~entries (Assumption M4). In \cite{RG_IEEE_12}, it is shown that certain rate-sharing matrices (which are not i.i.d.) can achieve distortions that are lower than the bounds given in Theorems 3 and 4. Therefore, a further contribution of this paper is that i.i.d.~matrices are strictly suboptimal in some problem settings.

\appendices


\section{Proof of Theorem~\ref{thm:LB_gen1}}\label{sec:thm:LB_gen1_proof}

The cornerstone of this proof is Fano's inequality which gives a lower bound on the error probability for any possible recovery algorithm in terms of the mutual information between $S^*$ and the pair $(\bY, \bA)$. We assume that the tuple $(D,p_X,\snr)$ is known throughout the system.

\begin{lemma}[Fano's Inequality]
 Let $S^*$ be distributed uniformly over all subsets of $[n]$ of size $k<n/2$. If $S^* \rw (\bY,\bA) \rw \hat{S}$ 
forms a Markov chain then
\begin{align}
\Pr[d(S^*,\hat{S}) > D] \ge 1 - \frac{I(S^*;\bY,\bA) + \log(2)}{\log {n \choose k} - \log \Big(\sum_{\ell = 0}^{\lceil D k \rceil }{k \choose \ell} {n-k \choose \ell} \Big)} \label{eq:fano}
\end{align}
for all $0 \le D \le 1$.
\end{lemma}
\begin{proof}
We follow the proof of Fano's inequality given in \cite{ElementsofIT} with some modifications to handle our error criterion. To begin, we define the random variable 
$$E = \begin{cases} 1, & \text{if $d(S^*,\hat{S}) > D$}\\
0, &\text{if $d(S^*,\hat{S}) \le D$}
\end{cases},
$$
and note that $\Pr[E=1] =  \Pr[d(S^*,\hat{S}) > D].$

Using the chain rule for entropy, $H(E,S^* | \bY,\bA, \hat{S})$ can be written two ways as
\begin{align*}
H(E,S^* | \bY,\bA, \hat{S}) &= H(S^*|\bY,\bA,\hat{S}) + H(E | S^*, \bY, \bA,\hat{S})\\
& = H(E |\bY,\bA,\hat{S} ) + H(S^*|E, \bY, \bA,\hat{S}).
\end{align*}
By the Markov property, $H(S^*|\bY,\bA,\hat{S}) = H(S^*|\bY,\bA)$. Since entropy is nonnegative, $H(E | S^*, \bY, \bA,\hat{S}) \ge 0$. Also, since conditioning cannot increase entropy, $H(E |\bY,\bA,\hat{S})\le H(E) \le \log(2)$ and $H(S^* |E,\bY,\bA,\hat{S})\le H(S^*|E,\hat{S})$. Putting everything together we obtain
\begin{align}
 H(S^*|\bY,\bA) - \log 2 &\le H(S^*|E, \hat{S}) \\
 & = \Pr[E = 1] H(S^* |E=1, \hat{S}) \nonumber \\
 & \quad + \Pr[E=0] H(S^*|E =0, \hat{S})\label{eq:H_fano_c}
\end{align}

Since the uniform distribution maximizes the entropy of $S^*$, 
\begin{align}
H(S^* |E = 1, \hat{S}) \le  \log {n \choose k}.\label{eq:H_fano_d}
\end{align}
Also, since the distortion measure $d(\cdot,\cdot)$ corresponds to the maximum of the two detection error rates, we may assume without any loss of generality that $\hat{S}$ has cardinality $k$. Therefore, a simple counting argument gives
\begin{align}
H(S^* | E = 0, \hat{S}) \le  \log \bigg( \sum_{\ell = 0}^{\lfloor D k \rfloor} {k \choose \ell} {n-k \choose \ell} \bigg).\label{eq:H_fano_e}
\end{align}
Plugging \eqref{eq:H_fano_d} and \eqref{eq:H_fano_e} back into \eqref{eq:H_fano_c} and solving for the error probability $\Pr[E=1]$ completes the proof. 
\end{proof}

The next step in the proof is to verify that the right-hand side of \eqref{eq:fano} is bounded away from zero for all sequences of problems obeying the assumptions of Theorem~\ref{thm:LB_gen1}. For each problem of size $n$, let $k = \lceil \kappa\, n \rceil$ where the dependence on $n$ is implicit. Using Stirling's approximation \cite[Lemma 17.5.1]{ElementsofIT}, it is straightforward to verify that 
\begin{align}
\lim_{n \rw \infty}  \frac{1}{n} \bigg[ \log {n \choose k}-  \log \bigg(\sum_{\ell = 0}^{\lceil D k \rceil }{k \choose \ell} {n-k \choose \ell} \bigg) \bigg] = R(D;\kappa) \label{eq:R_D_lim}
\end{align}
where $R(D;\kappa)$ is given in \eqref{eq:R_D}.

Combining \eqref{eq:fano} and \eqref{eq:R_D_lim} it follows that a distortion $D$ is not achievable if 
\begin{align}
\limsup_{n \rw \infty} \frac{1}{n} I(S^*; \bY, \bA) < R(D;\kappa). \label{eq:nec_bound}
\end{align}

The remainder of the proof is dedicated to upper bounding the left-hand side of \eqref{eq:nec_bound}. Starting with the chain rule for mutual information, we have
\begin{align}
I(S^*;\bY,\bA) &= I(S^*;\bY|\bA) + I(S^*;\bA)\\
 &= I(S^*;\bY|\bA) \label{eq:I_bound_M3}\\
 & \le I(\bX; \bY|\bA) \label{eq:I_bound_data_process}
\end{align}
where \eqref{eq:I_bound_M3} follows from the independence of Assumption M3 and \eqref{eq:I_bound_data_process} follows from the data processing inequality and the fact that $S^* \rw \bX \rw \bY$ forms a Markov chain.

Next, we can write
\begin{align}
I(\bX;\bY|\bA = A)& = I(\bX; A\bX + \snr^{-1/2} \bW) \nonumber \\
&   =I\big(\bX - \bE[\bX]; A(\bX-\bE[\bX])+ \snr^{-1/2} \bW\big) \nonumber \\
&   \le \max_{Z} I(A \bZ;A \bZ + \snr^{-1/2} \bW) \label{eq:info_Z}
\end{align}
where the maximum is over all $n$-dimensional random vectors $\bZ$ obeying the power constraint
\begin{align}
\bE[\bZ \bZ^T]& = \bE\big[(\bX-\bE[\bX]) (\bX-\bE[\bX])^T\big]= V_X I_{n \times n}.
\end{align}
It is well known (see e.g.~\cite{ElementsofIT}) that the maximum of \eqref{eq:info_Z} is attained when the entries of $\bZ$ are i.i.d.~$\mathcal{N}(0,V_X)$, and thus we obtain
\begin{align}
I(\bX;\bY|\bA = A)& \le  \frac{1}{2} \log \det ( I_{m \times m} + \snr \, V_X AA^T). \label{eq:I_cond_A}
\end{align}

By the concavity of the log determinant, Hadamard's inequality, and Jensen's inequality we can bound the expectation of \eqref{eq:I_cond_A} with respect to a random matrix $\bA$ obeying the normalization of Assumption M3 as follows:
\begin{align}
& \bE\Big[\frac{1}{2} \log \det ( I_{m \times m} + \snr \, V_X \bA\bA^T) \Big] \nonumber\\
& \le  \frac{1}{2} \log \det\big ( I_{m \times m} + \snr \, V_X \bE\big[ \bA\bA^T]\big) \nonumber \\
& = \frac{m}{2} \log(1 + \snr\,V_X) \label{eq:I_gen_jen_a}.
\end{align}

Alternatively, starting with Sylvester's determinant theorem, we can write
\begin{align}
& \bE\Big[\frac{1}{2} \log \det ( I_{m \times m} + \snr \, V_X \bA\bA^T) \Big]\nonumber\\
  &=
 \bE\Big[\frac{1}{2} \log \det ( I_{n \times n} + \snr \, V_X \bA^T\bA) \Big] \nonumber \\
& \le  \frac{1}{2} \log \det\big ( I_{n \times n} + \snr \, V_X \bE\big[ \bA^T \bA]\big) \nonumber \\
& = \frac{n}{2} \log\Big(1 +  \frac{m}{n} \snr\,V_X\Big).\label{eq:I_gen_jen_b}
\end{align}
Combining \eqref{eq:I_cond_A}, \eqref{eq:I_gen_jen_a}, and \eqref{eq:I_gen_jen_b} gives
\begin{align*}
&I(\bX;\bY|\bA) \\
&\le  \min \bigg( \frac{m}{2} \log(1 + \snr\,V_X), \frac{n}{2} \log\Big(1 +  \frac{m}{n} \snr\,V_X\Big)\bigg),
\end{align*}
and hence
\begin{align}
&\limsup_{n \rw \infty} \frac{1}{n} I(S^*; \bY, \bA) \nonumber\\
&< \frac{\min(1,\rho)}{2} \log\big(1+ \max(1,\rho) V_X\, \snr\big), \label{eq:nec_bound_2}
\end{align}
for any sequence of matrices obeying Assumptions M1-M3. Combining \eqref{eq:nec_bound} and \eqref{eq:nec_bound_2} completes the proof of Theorem~\ref{thm:LB_gen1}.

\section{Proof of Lemma~\ref{prop:genie}}\label{sec:prop_genie_proof}

This proof is based on a genie argument. Suppose that a genie provides the recovery algorithm with the pair $(G,\bX_{G})$ where $G$ is a subset of the sparsity pattern $S^*$ and $\bX_{G}$ is a $|G|$-dimensional vector corresponding to the entries of $\bX$ indexed by $G$. Given this extra information, the recovery algorithm must then determine which of the remaining unknown entries $\{X_i : i \notin G\}$ are nonzero. Clearly, any lower bound on the achievable distortion $D$ in the genie-aided setting is also a lower bound on the achievable distortion in the original setting. 

In the following sections, we first describe how the genie selects the index set $G$. We then show that the resulting recovery problem is equivalent to the original recovery problem with altered parameters.

\subsection{Genie Selection Strategy}

The set $G$ is constructed as follows: each index $i =1,2,\cdots,n$ is reported, independently of the other indices, with probability $q(X_i)$ where the function $0 \le q(x) \le 1$ is chosen such that for all $t \in \mathbb{R}$, $$\Pr[ X_i \le t | \text{ $i$ is not reported}] = \Pr[ Z \le t]$$
where $Z \sim p_Z$. By the constraints \eqref{eq:kZ_bound} and \eqref{eq:pZ_bound} it can be verified that the function $q(x)$ exists and that $q(0) = 0$. In words, the genie ``prunes'' the entries of $\bX$ in a way such that the unreported entries are marginally distributed according to the distribution $p_Z$.

We now make several observations. First, since $q(0) = 0$, only nonzero entries are reported and so $G \subseteq S^*$. Second, since the indices are selected independently, the remaining nonzero entries $\{X_i : i  \in S^*\backslash G\}$ are i.i.d.~according to the nonzero part of $p_Z$. Finally, conditioned on the cardinality $|G|$, the set $S^*\backslash G$ is distributed uniformly over all subsets of $[n]\backslash G$ of size $|S^*| - |G|$.

As a consequence of the above observations, the sequence of vectors corresponding to $\bX_{[n] \backslash G}$ satisfies Assumptions SS1-SS2 with distribution $p_Z$. Moreover, if we let $\tilde{\bY}$ denote the measurements corresponding to the vector $\bX_{[n]\backslash G}$ and measurement matrix $\bA_{[n]\backslash G}$, i.e. 
\begin{align}
\tilde{\bY} &
= \bA_{[n]\backslash G} \bX_{[n]\backslash G}  + \frac{1}{\sqrt{\snr}} \bW, \label{eq:model_genie}
\end{align} 
then it is straightforward to show that an appropriately normalized version of the measurement model given by \eqref{eq:model_genie} obeys Assumptions M1-M3 with sampling rate $\tilde{\rho}$ and signal-to-noise ratio $\tilde{\snr}$.

\subsection{Lower Bound on Genie-Aided Recovery}

We now derive a necessary condition for recovery in the genie-aided setting. We begin with the following key fact: if the set $G$ is chosen according to the selection strategy outlined above, the tuple $(\tilde{\bY},\bA, G)$ is a sufficient statistic for estimation of $S^*$. To see why, observe that
\begin{align}
&I(S^*; \bY,\bA, G, \bX_G )\nonumber\\
&=I(S^*; \bY - \bA_G \bX_G ,\bA, G, \bX_G ) \label{eq:I_genie_a1}\\
 &=I(S^*; \tilde{\bY},\bA, G, \bX_G ) \label{eq:I_genie_a}\\
&=I(S^*; \tilde{\bY},\bA,G )  + I(S^*; \bX_G |\tilde{\bY},\bA ,G ) \label{eq:I_genie_b} \\
&=I(S^*; \tilde{\bY},\bA, G)  \label{eq:I_genie_c}
\end{align}
where: \eqref{eq:I_genie_a} follows from the definition of $\tilde{Y}$; \eqref{eq:I_genie_b} follows from the chain rule for mutual information; and 
\eqref{eq:I_genie_b} follows from the fact that $S^*$ and $\bX_G$ are conditionally independent given the pair $(\tilde{\bY},\bA,G)$. 

Let $\hat{S}$ denote the optimal estimate of the sparsity pattern in the genie-aided setting (i.e.~the sparsity pattern estimate that minimizes the error probability). By the arguments above, we know that
\begin{align}
S^* \rw (\tilde{\bY},\bA,G) \rw \hat{S} \label{eq:markovSG}
\end{align}
forms a Markov chain. Also, by the optimality of $\hat{S}$ and the fact that distortion measure $d(\cdot,\cdot)$ corresponds to the maximum of the two detection error rates, it can also be shown that $\hat{S}$ contains the set $G$ and has the same cardinality as $S^*$. Therefore, the sparsity pattern distortion can be expressed as
\begin{align}
d(S^*,\hat{S}) =\Big( \frac{|S^*| - |G|}{|S^*|}\Big) d(S^* \backslash G , \hat{S} \backslash G). \label{eq:d_SG}
\end{align}
Note that
\begin{align}
\lim_{n \rw \infty} \Big( \frac{|S^*| - |G|}{|S^*|}\Big) = \Big( \frac{1-\kappa_X}{\kappa_X}\Big) \Big( \frac{\kappa_Z}{1-\kappa_Z}\Big) \label{eq:limSG}
\end{align}
almost surely under Assumptions SS1-SS2.

We now arrive at the crux of the argument. Suppose that the distortion $\tilde{D}$ is not achievable for the tuple $(\tilde{\rho}, p_Z, \tilde{\snr})$. By \eqref{eq:markovSG} and the fact that the observation model given in \eqref{eq:model_genie} corresponds to the tuple $(\tilde{\rho}, p_Z, \tilde{\snr})$, it follows that the error probability $$\Pr[ d(S^* \backslash G , \hat{S} \backslash G) \ge \tilde{D}]$$ corresponding to the genie-aided setting is bounded away from zero for all $n$. By \eqref{eq:d_SG} and \eqref{eq:limSG}, it then follows that the distortion $D$ is not achievable for the tuple $(\rho, p_X, \snr)$. This concludes the proof of Lemma~\ref{prop:genie}.


\section{Proof of Theorem~\ref{thm:LB_iid_1}}\label{sec:thm:LB_iid_1_proof}

One weakness of the proof of Theorem~\ref{thm:LB_gen1} is that the data processing inequality used to upper bound the mutual information $I(S^*;\bY|\bA)$ in \eqref{eq:I_bound_data_process} is not tight. In this proof, we derive a stronger upper bound that takes into account the fact that the values of the nonzero elements are unknown. We assume throughout the proof that the nonzero entropy power $N_X$ is strictly positive. 

Using the chain rule for mutual information, $I(S^*,\bX;\bY|\bA)$ can be written two ways as
\begin{align*}
I(S^*,\bX;\bY|\bA) & = I(S^*;\bY|\bA) + I(\bX;\bY|S^*,\bA)\\
& =  I(\bX;\bY,\bA) +I(S^*;\bY|\bX,\bA).
\end{align*}
Since $\bS \rw \bX \rw (\bY,\bA)$ forms a Markov chain, the mutual information $I(S^*;\bY|\bX,\bA)$ is equal to zero and 
\begin{align}
I(S^*;\bY|\bA) = I(\bX;\bY|\bA) - I(\bX;\bY|S^*,\bA).\label{eq:chain_rule}
\end{align}
Conceptually, the term $I(\bX;\bY|S^*,\bA)$ quantifies the amount of $I(\bX;\bY|\bA)$ that is ``used up'' describing the values of the nonzero elements, and hence cannot contribute to estimation of the sparsity pattern. 

Following the proof of Theorem~\ref{thm:LB_gen1}, the first term on the right-hand side of \eqref{eq:chain_rule} can be upper bounded as
\begin{align}
I(\bX;\bY|\bA) \le \frac{1}{2}\bE\Big[ \log \det \big( I_{m \times m} + \snr\, V_X \bA\bA^T\big)\Big] \label{eq:I_A_B1}
\end{align}
where the expectation is taken with respect to the random matrix $\bA$. 

To deal with the second term on the right-hand side of \eqref{eq:chain_rule} we first consider the case $m  \le k$. If we let 
\begin{align}
N(\bZ) = \frac{1}{2 \pi e} \exp\Big( \frac{2}{m} h(\bZ)\Big)
\end{align}
denote the entropy power of an $m$-dimensional random vector $\bZ$, then it follows straightforwardly that
\begin{align}
&I(\bX;\bY|S^*=S,\bA=A)\nonumber \\
 & = I(\bX_S ;  \sqrt{\snr} A_S \bX_S + \bW) \nonumber \\
& = h(\sqrt{\snr} A_S \bX_S +  \bW) - h( \sqrt{\snr} A_S \bX_S + \bW | \bX_S) \nonumber \\
& = \frac{m}{2} \log \big( 2 \pi e\, N( \sqrt{\snr}A_S \bX_S + \bW)\big)  - \frac{m}{2} \log( 2 \pi e) \nonumber \\
& = \frac{m}{2} \log \big(  N(\sqrt{\snr} A_S \bX_S +  \bW) \big) \label{eq:I_N1}.
\end{align}

Using a generalization of the entropy power inequality \cite{zamir:1993}, we can write
\begin{align}
N(\sqrt{\snr} A_S \bX_S +  \bW) &\ge N(\sqrt{\snr} A_S \bX_S) + N(\bW)\\
& \ge\snr\, \Big(\frac{  N_X}{\kappa} \Big) \det( A_S A_S^T)^{1/m} + 1, \label{eq:EPI_1}
\end{align}
where $N_X = \kappa N(X_i | i \in S^*)$ denotes the nonzero entropy power of $p_X$. Note that the assumption $m \le k$ is critical here since the determinant $A_S A_S^T$ is equal to zero for all $m <k$. 

Plugging \eqref{eq:EPI_1} back into \eqref{eq:I_N1} leads to
\begin{align}
&I(\bX;\bY|S^*,\bA) \nonumber\\ & \ge \frac{m}{2} \bE \Big[ \log\big( 1+\snr\,  N_X\, \kappa^{-1} \det( \bA_{S^*} \bA_{S^*}^T)^{1/m}  \big) \Big] \label{eq:I_A_B2}
\end{align}
where the expectation is with respect to the random matrix $\bA_{\bS^*}$.

Next we consider the case $m > k$. If the matrix $A_S$ is full rank and we let $A_{S}^\dagger$ denote its Moore-Penrose pseudoinverse, we can write
\begin{align}
&I(\bX;\bY|S^*=S,\bA=A)\nonumber\\ 
 & = I(\bX_S ;  \sqrt{\snr} \, \bX_S + A^\dagger_S \bW) \nonumber \\
& = h(\sqrt{\snr} \, \bX_S +  A^\dagger_S \bW) - h( \sqrt{\snr}  \bX_S + A^\dagger_S \bW | \bX_S) \nonumber\\
& = \frac{k}{2} \log \big(  N(\sqrt{\snr} \, \bX_S +  A^\dagger_S  \bW) \big)  + \frac{1}{2} \log\det( A_S^T A_S) \nonumber \\
&\ge  \frac{k}{2} \log\big( 1+\snr\,  N_X\, \kappa^{-1} \det( A_{S}^T A_{S})^{1/k}  \big) \label{eq:I_N2}
\end{align}
where \eqref{eq:I_N2} follows again from the entropy power inequality. Thus, we obtain 
\begin{align}
&I(\bX;\bY|S^*,\bA)  \nonumber \\
& \ge \frac{k}{2} \bE \Big[ \log\big( 1+\snr\,  N_X\, \kappa^{-1} \det( \bA^T_{S^*} \bA_{S^*})^{1/k}  \big) \Big], \label{eq:I_A_B3}
\end{align}
where the expectation is with respect to the random matrix $\bA_{\bS^*}$.  

Finally, to characterize the asymptotic behavior of the bounds in \eqref{eq:I_A_B1}, \eqref{eq:I_A_B2}, and \eqref{eq:I_A_B3}, we use use the fact that the spectral distributions of the matrices $\bA$ and $\bA_S$ converge to a non-random limit known as the Marcenko--Pastur Law (see Appendix~\ref{sec:random_matrix}).

Combining Lemma~\ref{lem:MP_exp} in Appendix~\ref{sec:random_matrix} with the upper bound \eqref{eq:I_A_B1} leads immediately to
\begin{align}
\limsup_{n \rw \infty} \frac{1}{n} I(\bX;\bY|\bA) \le \frac{1}{2} \mathcal{V}(\rho, V_X \, \snr).
\end{align}
Similarly, combining Lemma~\ref{lem:MP_geom}  in Appendix~\ref{sec:random_matrix} with the lower bounds \eqref{eq:I_A_B2} and \eqref{eq:I_A_B3} leads to
\begin{align}
\liminf_{n \rw \infty} \frac{1}{n} I(\bX;\bY|S^*, \bA) \ge \frac{1}{2} \kappa \mathcal{V}_{LB}(\rho/\kappa , N_X \, \snr)
\end{align}
where $\mathcal{V}_{LB}(r,\gamma)$ is given by \eqref{eq:V_LB_funct}. Plugging these limits back into \eqref{eq:chain_rule} and \eqref{eq:nec_bound} completes the proof of Theorem~\ref{thm:LB_iid_1}.

\section{Asymptotic Spectral Convergence} \label{sec:random_matrix}

This appendix states two useful results from random matrix theory and gives bounds on the functions $\mathcal{V}(r,\gamma)$ and $\mathcal{V}_\text{LB}(r,\gamma)$ introduced Theorem~\ref{thm:LB_iid_1}.

\begin{lemma}\label{lem:MP_exp} {\cite{VerSha99}} Let $\bA$ denote an $m \times n$ random matrix whose entries are i.i.d.~with mean zero and variance $1/n$. If $m/n \rw r$ as $n \rw \infty$, then
\begin{align}
\lim_{n \rw \infty} \frac{1}{n} \log \det \big( I_{m\times m} + \gamma \bA\bA^T  \big) = \mathcal{V}( r,\gamma ) \label{eq:V_a}
\end{align}
almost surely where $\mathcal{V}(r,\gamma)$ is given by \eqref{eq:V_funct}.
\end{lemma}

\begin{lemma}\label{lem:MP_geom} {\cite{Salo06}} Let $\bA$ denote an $m \times n$ random matrix whose entries are i.i.d.~with mean zero and variance $1/n$. If $m/n \rw r$ as $n \rw \infty$, then
\begin{align}
\lim_{n \rw \infty} \big( \det(\bA \bA^T)\big)^{1/m} &= \Big( \frac{1}{1-r} \Big)^{1/r- 1} \frac{1}{e},  &\text{if $r < 1$}\\
\lim_{n \rw \infty} \big( \det(\bA^T \bA)\big)^{1/n} &=  \frac{1}{e},  &\text{if $r= 1$}\\
\lim_{n \rw \infty} \big( \det(\bA^T \bA)\big)^{1/n} &= \Big( \frac{r}{r-1} \Big)^{r- 1} \frac{1}{e},  &\text{if $r> 1$}
\end{align}
almost surely.
\end{lemma}

Under the assumptions of Lemma~\ref{lem:MP_geom}, it thus follows that
\begin{align}
\lim_{n \rw \infty} r \log\big(1+   \gamma \det(\bA\bA^T)^{1/m}\big) = \mathcal{V}_{LB}(r,\gamma),\quad r \le 1 \label{eq:VLB_a}
\end{align}
and
\begin{align}
\lim_{n \rw \infty} \log\big(1+ r  \gamma \det(\bA\bA^T)^{1/n}\big) = \mathcal{V}_{LB}(r,\gamma), \quad r > 1.\label{eq:VLB_b}
\end{align}

The functions $\mathcal{V}(r,\gamma)$ and $\mathcal{V}_\text{LB}(r,\gamma)$ obey the following series of inequalities:
\begin{align}
r \log(1+r \gamma) 
& \ge \min(1,r) \log\big(1+ \max(1,r) \gamma\big) \label{eq:Vinq_a}\\
& \ge \mathcal{V}(r,\gamma)    \label{eq:Vinq_b}\\
& \ge \mathcal{V}_\text{LB}(r,\gamma)  \label{eq:Vinq_c}\\
& \ge  \min(1,r) \log\big(1+ \max(1,r) \gamma / e\big), \label{eq:Vinq_d}
\end{align}
where \eqref{eq:Vinq_a} and \eqref{eq:Vinq_b} follow from the concavity of the logarithm and Jensen's inequality and \eqref{eq:Vinq_c} follows from \eqref{eq:V_a}, \eqref{eq:VLB_a}, \eqref{eq:VLB_b}, and Hadamard's inequality.

The next result shows that functions $\mathcal{V}(r,\gamma)$ and $\mathcal{V}_\text{LB}(r,\gamma)$ behave similarly when $\gamma$ is large.

\begin{lemma} \label{lem:VLB_comp}
For any $r < 1$, 
\begin{align}
\lim_{\gamma \rw \infty}\max_{0 \le s \le r} \Big| \mathcal{V}(s,\gamma) - V_\text{LB}(s,\gamma)\Big| = 0.
\end{align}
\end{lemma}
\begin{proof}
With a bit of algebra, it can be verified that
\begin{align}
&\gamma r - \mathcal{F}(r,\gamma) \nonumber \\
& = \frac{1}{2} \Big[ \gamma(1-r) \Big( \sqrt{1 + \tfrac{2}{\gamma(1-r)} \left( \tfrac{1+r}{1-r} + \tfrac{1}{2 \gamma(1-r)} \right)} -1 \Big) -1 \Big] \nonumber\\
& \le \frac{r}{1-r} + \frac{1}{2 \gamma (1-r)} \label{eq:Vbounda}
\end{align}
where \eqref{eq:Vbounda} follows from the bound $\sqrt{1+2 x} \le 1+x$. Plugging this inequality back into the definition of $\mathcal{V}(r,\gamma)$ gives an upper bound $\mathcal{V}_\text{UB}(r,\gamma)$. 
At this point it is straightforward to verify that
\begin{align*}
&\lim_{\gamma \rw \infty}\max_{0 \le s \le r} \Big| \mathcal{V}_\text{UB}(s,\gamma) - V_\text{LB}(s,\gamma)\Big| \\
& = \lim_{\gamma \rw \infty} \Big| \mathcal{V}_\text{UB}(r,\gamma) - V_\text{LB}(r,\gamma)\Big| = 0,
\end{align*}
which completes the proof.
\end{proof}

\section{Proofs of Low-Distortion Behavior} \label{sec:low_D_appendix}

\subsection{Proof of Corollary~\ref{prop:lowD_SNR}}

For this proof we begin with the bound in Corollary~\ref{cor:LB_gen2a} evaluated with $D' = \min(1,\alpha D)$. For all $D < 1/\alpha$, this gives
\begin{align}
\rho^* &\ge \frac{ 2(1-\kappa + \alpha \kappa D) R(\frac{1}{\alpha}, \frac{\kappa \alpha D}{1 - \kappa + \alpha \kappa D})}{\log(1+P(\alpha D;p_X)\,\snr)}\\
& \ge \frac{ 2(1-\kappa + \alpha \kappa D) R(\frac{1}{\alpha}, \frac{\kappa \alpha D}{1 - \kappa + \alpha \kappa D})}{P(\alpha D;p_X)\,\snr}, \label{eq:rho_Low_D_b}
\end{align}    
where \eqref{eq:rho_Low_D_b} follows from the bound $\log(1+x) \le x$.

Next, we consider the numerator in \eqref{eq:rho_Low_D_b}. Observe that
\begin{align}
& R\Big(\frac{1}{\alpha}, \frac{\kappa \alpha D}{1 - \kappa + \alpha \kappa D}\Big) \nonumber\\
 & = H\Big(\frac{\kappa \alpha D}{1 - \kappa + \alpha \kappa D}\Big) - \Big(\frac{\kappa \alpha D}{1 - \kappa + \alpha \kappa D}\Big) H\Big(\frac{1}{\alpha}\Big)\nonumber\\
 & \quad - \Big(\frac{1-\kappa}{1 - \kappa + \alpha \kappa D}\Big) H\Big(\frac{\kappa D}{1-\kappa}\Big).
\end{align}
Using the fact that, for any constant $c > 0$,
\begin{align}
\lim_{p \rw 0} \frac{H(c\,p)}{p \log(1/p)} = c,
\end{align}
it thus follows that
\begin{align}
\lim_{D \rw 0} \frac{(1-\kappa + \alpha \kappa D) R(\frac{1}{\alpha}, \frac{\kappa \alpha D}{1 - \kappa + \alpha \kappa D})}{D\log(1/D)} = (\alpha -1) \kappa. \label{eq:R_lim}
\end{align}
Plugging \eqref{eq:R_lim} back into \eqref{eq:rho_Low_D_b} completes the proof. 

\subsection{Proof of Corollary~\ref{cor:lowD}}

We begin with distributions that are bounded away from zero. By the definition of $P(D;\kappa)$, it follows straightforwardly that
\begin{align}
P(\alpha D;p_X) \ge \alpha \kappa  D B^2 \label{eq:PLB_bounded}
\end{align}
for all distributions $p_X$ in the bounded class $\mathcal{P}_\text{Bounded}(\kappa,B)$. Combining \eqref{eq:PLB_bounded} with Corollary~\ref{prop:lowD_SNR} gives
\begin{align}
\liminf_{D \rw 0} \frac{\rho^*}{ \log(1/D)} \ge \Big(\frac{\alpha -1}{\alpha}\Big) \frac{2}{B^2 \cdot \snr}.
\end{align}
Since $\alpha > 1$ is arbitrary, the leading term $(\alpha -1)/\alpha$ can be made arbitrarily close to one.

Next we consider distributions with polynomial decay. In \cite[Eq. (215)]{RG_IEEE_12}, it is shown that
\begin{align}
\lim_{D \rw 0} \frac{P(\alpha D;p_X)}{D^{1+2/L}} = \alpha^{1+2/L} \cdot \frac{ \kappa\, \tau^{-2/L}}{1+2L} \label{eq:PLB_poly}
\end{align}
for all distributions $p_X$ in the polynomial decay class $\mathcal{P}_\text{Poly.}(\kappa,L,\tau)$. Combining \eqref{eq:PLB_poly} with Corollary~\ref{prop:lowD_SNR} gives
\begin{align}
\liminf_{D \rw 0} \frac{\rho^*}{D^{2/L} \log(1/D)} \ge \Big(\frac{\alpha -1}{\alpha^{1+2/L}}\Big) \frac{2  (1+2/L)}{\tau^{-2/L}\cdot \snr}. \label{eq:PLB_polyb}
\end{align}
Since $\alpha > 1$ is arbitrary, the leading term on the right-hand side of \eqref{eq:PLB_polyb} can be optimized by choosing 
$\alpha = 1+L/2$. This completes the proof.

\section{Proofs of High-SNR Behavior} \label{sec:high_SNR_appendix}

\subsection{Proof of Corollary~\ref{thm:high_SNR_1}}

For this result, we compute the infinite SNR limit of the left-hand side of \eqref{eq:LB_iid_1}. Since, the achievable distortion is non-increasing in the SNR, this limit gives a valid lower bound for any SNR. Using Lemma~\ref{lem:VLB_comp} in Appendix~\ref{sec:random_matrix}, we have
\begin{align*}
&\lim_{\snrs \rw \infty} \mathcal{V}(\rho,V_X\,\snr) - \kappa \mathcal{V}_{LB}(\rho/\kappa, N_X\, \snr) \nonumber\\
& = \lim_{\snrs \rw \infty} \mathcal{V}_\text{LB}(\rho,V_X\,\snr) - \kappa \mathcal{V}_{LB}(\rho/\kappa, N_X\, \snr)\\
& = \rho \log \Big( \frac{V_X}{N_X} \Big)  + (1-\rho) \log \big( \tfrac{1}{1-\rho} \big) - (\kappa - \rho) \log \big( \tfrac{\kappa}{\kappa - \rho} \big)
\end{align*}
where we have used the fact that $\rho$ and $\rho/\kappa$ are both less than one. 

\subsection{Proof of Corollary~\ref{thm:high_SNR_2}}
Similar to the proof of Corollary~\ref{thm:high_SNR_1}, we study the high SNR behavior of the left-hand side of \eqref{eq:LB_iid_1}. To begin, let $(D,p_X)$ be a fixed pair satisfying \eqref{eq:DinfLim}. For each $\gamma \ge 0$, let $\rho_\gamma$ denote the unique solution to the fixed point equation:
\begin{align}
\mathcal{V}(\rho_\gamma,V_X\, \gamma) - \kappa \mathcal{V}_\text{LB}(\rho_\gamma/\kappa,N_X\, \gamma) = 2 R(D;\kappa).\label{eq:fixpoint} 
\end{align}
Clearly, $\rho_\gamma$ gives a lower bound on the sampling rate distortion function $\rho^*$ evaluated with $\snr  =\gamma$. 

We are interested in the behavior of $\rho_\gamma$ as $\gamma$ becomes large. By Corollary~\ref{thm:high_SNR_1} it follows that $\rho_\gamma \ge \kappa$ for all $\gamma$. By inspection of the left-hand side of \eqref{eq:fixpoint}, it also follows that $\limsup_{\gamma \rw \infty}\rho_\gamma \le \kappa$, since otherwise the left-hand side would increase without bound as $\gamma \rw \infty$. Therefore, we conclude that
\begin{align}
\rho_\gamma = \kappa + o(\gamma), \label{eq:r_gamma}
\end{align}
where, for a function $f(x)$, the notation $f(x)=o(x)$ means that $\lim_{x \rw \infty} f(x) = 0$.

Now, starting with the first term on the left-hand side of \eqref{eq:fixpoint}, we can write
\begin{align}
&\mathcal{V}(\rho_\gamma, V_X\, \gamma)  \nonumber\\
&= \mathcal{V}_\text{LB}(\rho_\gamma, V_X\, \gamma) + o(\gamma) \label{eq:Vscaling_a}\\
&= \rho_\gamma \log(\gamma) + \rho_\gamma \log\big(V_X   \big( \tfrac{1}{1-\rho_\gamma} \big)^{1/\rho_\gamma- 1} \tfrac{1}{e}\big) + o(\gamma) \label{eq:Vscaling_b}\\
&= \rho_\gamma \log(\gamma) + \kappa \log\big(V_X   \big( \tfrac{1}{1-\kappa} \big)^{1/\kappa- 1} \tfrac{1}{e}\big) + o(\gamma)\label{eq:Vscaling_c}
\end{align}
where: \eqref{eq:Vscaling_a} follows from Lemma~\ref{lem:VLB_comp} in Appendix~\ref{sec:random_matrix}; \eqref{eq:Vscaling_b} follows from the definition of $\mathcal{V}_\text{LB}(r,\gamma)$ and the fact that $\rho_\gamma$ is eventually less than one; and \eqref{eq:Vscaling_c} follows from \eqref{eq:r_gamma}.

Similarly, starting with the second term on the left-hand side of \eqref{eq:fixpoint}, we can write
\begin{align}
&\kappa \mathcal{V}_\text{LB}(\tfrac{\rho_\gamma}{\kappa}, N_X\, \gamma)  \nonumber\\
&= \kappa \log(\gamma) + \kappa \log\big(N_X  \tfrac{\rho_\gamma}{\kappa} \big( \tfrac{\rho_\gamma}{\rho_\gamma-\kappa} \big)^{\rho_\gamma/\kappa- 1} \tfrac{1}{e}\big) + o(\gamma) \label{eq:VLBscaling_b}\\
&= \kappa \log(\gamma) + \kappa \log\big(N_X  \tfrac{1}{e}\big) + o(\gamma)\label{eq:VLBscaling_c}
\end{align}
where \eqref{eq:VLBscaling_b} follows from the definition of $\mathcal{V}_\text{LB}(r,\gamma)$ and the fact that $\rho_\gamma / \kappa > 1$, and  \eqref{eq:VLBscaling_c} follows from \eqref{eq:r_gamma}.

Plugging \eqref{eq:Vscaling_c} and \eqref{eq:VLBscaling_c} back into \eqref{eq:fixpoint} gives
\begin{align}
&(\rho_\gamma - \kappa) \log(\gamma)   + o(\gamma) \nonumber\\ 
&= 2R(D;\kappa) -\kappa \log\Big(\frac{V_X}{N_X}\Big) - (1-\kappa) \log\Big(\frac{1}{1-\kappa}\Big).
\end{align}
Since $\rho_\gamma$ is a lower bound on the sampling rate-distortion function, the proof is complete.


\section*{Acknowledgment}
We would like to thank Martin Wainwright for helpful discussions and pointers in early versions of this work and the anonymous reviewers for their helpful comments and suggestions. 

\bibliographystyle{ieeetran}



\begin{IEEEbiographynophoto}{Galen Reeves} received the B.S. degree in electrical and computer engineering from Cornell University in 2005 and the M.S. and Ph.D. degrees in electrical engineering and computer sciences from the University of California at Berkeley in 2007 and 2011 respectively. He is currently a postdoctoral scholar at Stanford university. His his research interests include compressed sensing, statistical signal processing, information theory, and machine learning.
\end{IEEEbiographynophoto}


\begin{IEEEbiographynophoto}{Michael Gastpar}
received the Dipl. El.-Ing. degree from ETH Z\"urich, in 1997, the M.S. degree from the University of Illinois at Urbana-Champaign, Urbana, IL, in 1999, and the
Doctorat \`es Science degree from Ecole Polytechnique F\'ed\'erale (EPFL), Lausanne, Switzerland, in 2002, all in electrical engineering. He was also a student in engineering and philosophy at the Universities of Edinburgh and Lausanne.

He is a Professor in the School of Computer and Communication
Sciences, Ecole Polytechnique F\'ed\'erale (EPFL), Lausanne, Switzerland.
He was an Assistant (2003-2008) and tenured Associate Professor (2008-2011)
with the Department of Electrical Engineering and Computer Sciences,
University of California, Berkeley, where he still holds a faculty position.
He also holds a faculty position at Delft University of Technology, The Netherlands,
and was a Researcher with the Mathematics of Communications Department,
Bell Labs, Lucent Technologies, Murray Hill, NJ.
His research interests are
in network information theory and related coding and signal processing techniques,
with applications to sensor networks and neuroscience.

Dr. Gastpar won the 2002 EPFL Best Thesis Award, an NSF CAREER Award
in 2004, an Okawa Foundation Research Grant in 2008, and an ERC Starting Grant in 2010.
He is an Information Theory Society Distinguished Lecturer (2009Ð2011). He was an Associate
Editor for Shannon Theory for the IEEE TRANSACTIONS ON INFORMATION
THEORY (2008--2011), and he has served as Technical Program Committee Co-Chair for the
2010 International Symposium on Information Theory, Austin, TX.
\end{IEEEbiographynophoto} 

\end{document}